\documentclass[final]{IEEEtran}

\usepackage{amsthm,amssymb,graphicx,multirow,amsmath,color,amsfonts}
\usepackage[update,prepend]{epstopdf}
\usepackage[noadjust]{cite}
\usepackage[latin1]{inputenc}
\usepackage{tikz}
\usetikzlibrary{arrows,calc}		
\usepackage{bbm} 
\usepackage{pdfpages}
\usepackage{flushend}
\usepackage{tabulary}
\usepackage{multirow}
\usepackage{comment}
\usepackage{subfig}
\allowdisplaybreaks
\newcommand*{\Scale}[2][4]{\scalebox{#1}{$#2$}}%
\newcommand{\PPP}[1]{\ensuremath{\Phi_{#1}}}
\newcommand{\h}[2]{\ensuremath{h_{{#2}}}}

\newcommand{\pow}[2]{\ensuremath{\bar{P}_{{#1}{#2}}}}

\def\Ri{\mathcal{R}_i}
\def\std{\ensuremath{\zeta}}
\def\mod{\ensuremath{{R}}}
\newcommand{\condA}[1]{\ensuremath{{\ncalA_{#1|{\ncalV}_0}}}}
\newcommand{\condAs}[1]{\ensuremath{{\ncalA_{#1|{\nu}_0}}}}

\def\matern{{Mat{\'e}rn}}
\def\lam{\ensuremath{\overline{\lambda}}}
\def\lamc{\ensuremath{\overline{\lambda'}}}
\def\serv{\ensuremath{W}}

\def\user{\ensuremath{\mathbf{Z}_u^{(i)}}}

\newcommand{\fservcond}[1]{\ensuremath{f_{{\serv}_{#1}|{\ncalV_0}}}}
\def\G{G(\alpha,\tau)}
\def\H{H(\alpha,\tau)}
\def\pcU{\pc_{j|\ncalV_0}^{(i),U}}
\def\pcL{\pc_{j|\ncalV_0}^{(i),L}}

\def\nb0{{\mathbf{0}}}
\def\nb1{{\mathbf{1}}}

\def\ncalA{{\mathcal{A}}}
\def\ncalB{{\mathcal{B}}}

\def\ncalF{{\mathcal{F}}}

\def\ncalH{{\mathcal{H}}}
\def\ncalI{{\mathcal{I}}}

\def\ncalK{{\mathcal{K}}}
\def\ncalL{{\mathcal{L}}}
\def\ncalM{{\mathcal{M}}}
\def\ncalN{{\mathcal{N}}}

\def\ncalP{{\mathcal{P}}}

\def\ncalR{{\mathcal{R}}}

\def\ncalV{{\mathcal{V}}}

\def\ncalZ{{\mathcal{Z}}}

\def\nbbE{{\mathbb{E}}}

\def\nbbP{{\mathbb{P}}}

\newtheorem{lemma}{Lemma}

\newtheorem{ndef}{Definition}

\newtheorem{theorem}{Theorem}
\newtheorem{prop}{Proposition}
\newtheorem{cor}{Corollary}

\newtheorem{remark}{Remark}

\def\figref#1{Fig.\,\ref{#1}}%

\def\pc{\mathtt{P_c}}

\def\sinr{\mathtt{SINR}}			

\def\sir{\mathtt{SIR}}

\usepackage{stackengine}

\def\delequal{\mathrel{\ensurestackMath{\stackon[1pt]{=}{\scriptstyle\Delta}}}}

\newcommand{\myeq}[1]{\mathrel{\overset{\makebox[0.07pt]{\mbox{(#1)}}}{=}}}
\graphicspath{{./Fig/}}
\title{Enriched $K$-Tier HetNet Model to Enable the Analysis of User-Centric Small Cell Deployments}

\author{Chiranjib Saha, Mehrnaz Afshang, and Harpreet S. Dhillon
\thanks{The authors are with Wireless@VT, Department of ECE, Virginia Tech, Blacksburg, VA, USA. Email: \{csaha, mehrnaz, hdhillon\}@vt.edu. The support of the US NSF (Grants CCF-1464293 and CNS-1617896) is gratefully acknowledged. This work was presented in part at the IEEE ICC in Kuala Lumpur, Malaysia, in May 2016 \cite{Saha1605:Downlink}.}
}

\begin{document}

\maketitle
\vspace{-.9cm}
\begin{abstract}
One of the principal underlying assumptions of current approaches to the analysis of heterogeneous
cellular networks (HetNets) with random spatial models is the uniform distribution of users independent
of the base station (BS) locations. This assumption is not quite accurate, especially for user-centric capacity-driven small cell deployments where low-power BSs are deployed in the areas of high user density, thus inducing a natural correlation in the BS and user locations. In order to capture this correlation,
we enrich the existing $K$-tier Poisson Point Process (PPP) HetNet model by considering user locations as Poisson Cluster Process (PCP) with the BSs at the cluster centers.  In particular, we provide the formal analysis of the downlink coverage probability in terms of a general density functions describing the locations of users around the BSs. The derived results are specialized for two cases of interest: (i) \textit{Thomas cluster process},  where the locations of the users around BSs are Gaussian distributed, and (ii) \textit{\matern\ cluster process}, where the users are uniformly distributed inside a disc of a given radius. Tight closed-form bounds for the coverage probability in these two cases are also derived. Our results demonstrate that the coverage probability decreases as the size of user clusters around BSs increases, ultimately collapsing to the result obtained under the assumption of PPP distribution of  users independent of the BS locations  when the cluster size goes to infinity. Using these results, we also handle {\em mixed} user distributions consisting of two types of users: (i) uniformly distributed, and (ii) clustered around certain tiers.   
\end{abstract}

\section{Introduction}\label{sec::intro}
Increasing popularity of Internet-enabled mobile devices, such as smartphones and tablets, has led to an unprecendented increase in the global mobile data traffic, which has in turn necessitated the need to dramatically increase the capacity of cellular networks. Not surprisingly, a key enabler towards increasing  network capacity at such a rate is to reuse spectral resources over space and time more aggressively.  This is already underway in the form of capacity-driven deployment of several types of low-power BSs in the areas of high user density, such as coffee shops, airport terminals, and downtowns of large cities \cite{3gpp, 3gpp2}. Due to the coexistence of the various types of low-power BSs, collectively called {\em small cells}, with the conventional high-power  macrocells, the resulting network is often termed as a heterogeneous cellular network (HetNet). Because of the increasing irregularity of BS locations in HetNets, random spatial models have become  preferred choice for the accurate modeling and tractable analysis of these networks. The most popular approach is to 
model the locations of different classes of BSs  by independent PPPs and perform the downlink analysis at a typical user chosen independent of the BS locations; see \cite{ElSHosJ2013,andrews2016primer,mukherjee2014analytical} and the references therein. 
However, none of the prior works has focused on developing tools for the more realistic case of  user-centric deployments in which the user and BS locations are correlated. 
Developing new tools to fill this gap is the main goal of this paper.
\subsection{Related Works}\label{subsec::relv}
Stochastic geometry has recently emerged as a useful tool for the analysis of cellular networks. Building on the single-tier cellular model developed in~\cite{andrews}, a multi-tier HetNet model was first developed in\cite{dhillon_tractable_2011, dhillonHetNet}, which was then extended in~\cite{xia, Mukherjee_Hetnet, MadhusudhananHetnet2016}. While the initial works were mainly focused 
on the downlink coverage and rate analyses, the models have since been extended in multiple ways, such as for load aware modeling of HetNets in \cite{dhillon_load-aware_2013}, traffic offloading in \cite{singh_offloading_2013}, and throughput optimization in \cite{cheung_throughput_2012}. 
Please refer to~\cite{ElSHosJ2013,Heath_hetnet_2013,andrews2016primer,mukherjee2014analytical,
elsawy2016modeling} for more pedagogical treatment of the topic as well as extensive surveys of the  prior art.
While PPP remains a popular abstraction of spatial distribution of cellular BSs randomly and independently coexisting over a finite but large area, a common assumption of the aforementioned analysis, as noted above, is that the users are uniformly distributed independent of the BS locations. However, in reality, the users form hotspots, which are where some types of small cells, such as picocells are deployed to enhance coverage and capacity \cite{jaziri_system_2016}. As a result, the user-centric deployment of small cells is one of the dominant themes in future wireless architectures \cite{boccardi_five_2014}. In such architectures, one can envision small cells being deployed to serve {\em clusters} of users. Such models are also being used by the standardization bodies, such as 3GPP~\cite{3gpp,3gpp2}. While there have been attempts to model such clusters of small cells by using PCP, e.g., see \cite{taranetz_performance_2012,ying_characterizing_2014, chen_downlink_2013,MultiChannel2013,HetPCPGhrayeb2015}, the user distribution is usually still assumed to be independent of the BS locations.

As noted above, modeling and performance analysis of user-centric capacity-driven deployment of small cells require accurate characterization of not only the spatial distribution of users but also correlation between the BS and user locations.  Existing works, however sparse, on the analysis of correlated non-uniform user distributions can be classified into two main directions. The first is to characterize the performance through detailed system-level simulations \cite{qualcom,lee_spatial_2014,mirahsahn,WangNonUnif2014}. As expected, the general philosophy is to capture the capacity-centric deployments by assuming higher user densities in the vicinity of small cell BSs, e.g., see \cite{qualcom}. In \cite{lee_spatial_2014}, the authors proposed non-uniform correlated traffic pattern generation over space and time based on log-normal or Weibull distribution.  On similar lines, \cite{mirahsahn} has introduced a low complexity PPP simulation approach for HetNets with correlated user and BS locations. 
 System level simulation shows that network performance significantly deteriorates with increased heterogeneity of users if there exists no correlation among the users and the small cell BS locations. But the HetNet performance  improves if the small cell BSs are placed at the cluster centers which are determined by means of clustering algorithms from a given user distribution \cite{WangNonUnif2014}.  
 
  The second direction, in which the contributions are even sparser, is to use analytic tools from stochastic geometry to characterize the performance of HetNets with non-uniform user distributions. One notable contribution in this direction is the generative model proposed in \cite{dh_non}, where non-uniform user distribution is generated from the homogeneous PPP by thinning the BS field independently, conditional on the active link from a typical user to its serving BS. While the resulting model is tractable, it suffers from two shortcomings: (i) it is restricted to single-tier networks and extension to HetNet is not straightforward, and (ii) even for single-tier networks, it does not allow the inclusion of any general non-uniform distribution of users in the model. In~\cite{li_mixed_2015}, the authors proposed a mixture of correlated and uncorrelated user distribution with respect to small cell BS  deployment and evaluated the enhancement in coverage probability as a function of correlation coefficient. Correlation has been introduced by generating users initially  as an independent PPP and later  shifting them towards the BSs with some probability. 
 In \cite{cheung_throughput_2012}, the authors have considered clustered users around femto-BSs as uniformly distributed on the circumference of a circle with fixed radius.
 Besides, some other attempts have been made at including non-uniform user distributions using simple models, especially in the context of indoor communications, e.g., see~\cite{urban}.  {In \cite{HaeDependence2015}, both the user and small scale BS locations are modeled as correlated \matern\ cluster processes having the same ``parent'' point process. But the analysis is simplified by  assuming  the distance between a user and its serving small cell BS either being fixed or a uniformly distributed random variable.} Overall, we are still somewhat short-handed when it comes to handling the analysis of user-centric deployments, which is the main focus of this paper.
With this brief overview of the prior art, we now discuss our contributions next.
\subsection{Contributions and Outcomes}
\subsubsection{New HetNet Model}
In this paper, we develop a new and more practical HetNet model  for accurately capturing the non-uniform user distribution as well as correlation between the locations of the users and BSs. In particular, the user locations have been modeled as superposition of PCPs. Correlation between the users and BSs under  user-centric capacity-driven deployment   has been captured by  assuming the BS locations as the parent point processes of the  cluster processes of users.  This  model is flexible enough to include any kind of user distribution around any arbitrary number of BS tiers as well as user distribution that is homogeneous and independent of the BS locations. This approach builds on our recent work on modeling device-to-device networks using PCPs \cite{AfsDhiJ2015a,AfsDhiJ2015}. 
\subsubsection{Downlink Analysis}
We derive exact expression for the coverage probability of a typical user chosen randomly from one of the clusters in this setup. 
The key step of 
our approach is the treatment of the cluster center  as an individual singleton 
tier. This enables the characterization of key distance distributions, which ultimately lead to easy-to-use expressions for the Laplace transform of interference distribution in all cases of interest. Using these components, we derive the   coverage probability of a randomly chosen user from one of the user clusters.  After characterizing the coverage probability under a general distribution of users, we  specialize our results for two popular PCPs, viz. Thomas and \matern\ cluster processes. Next, we provide  upper and lower bounds on coverage probability which are computationally more efficient than the exact expressions and reduce to closed form expressions for no shadowing when the user distribution is modelled as Thomas or \matern\ cluster process.
\begin{figure*}
  \centering
  \subfloat[Uniformly distributed users independent of the BS locations (prior art).\label{fig::deployment_uniform}]{\includegraphics[width=0.3\textwidth]{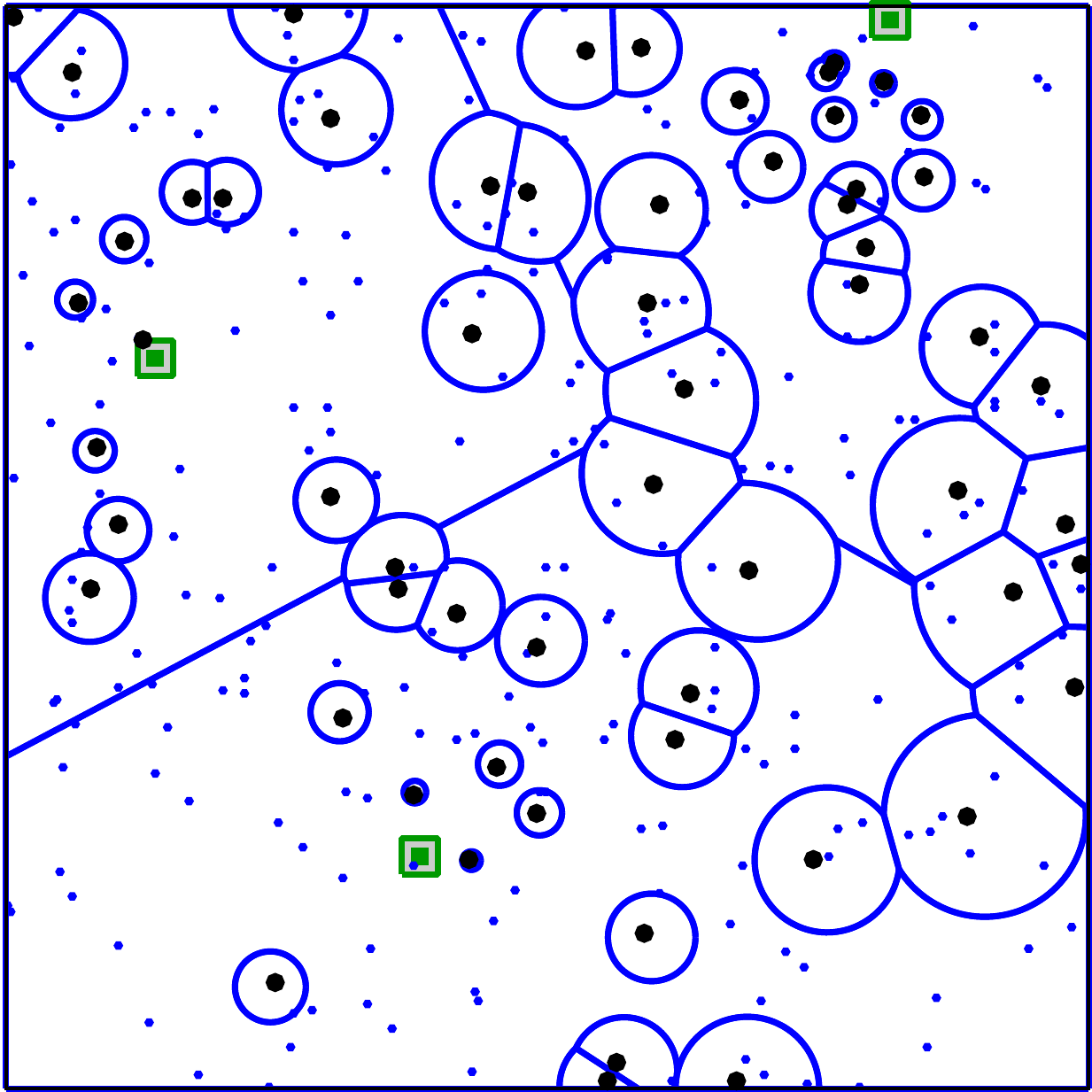}}\quad
  \subfloat[Users clustered around small cell BSs (this paper). \label{fig::deployment_thomas}]{\includegraphics[width=0.3\textwidth]{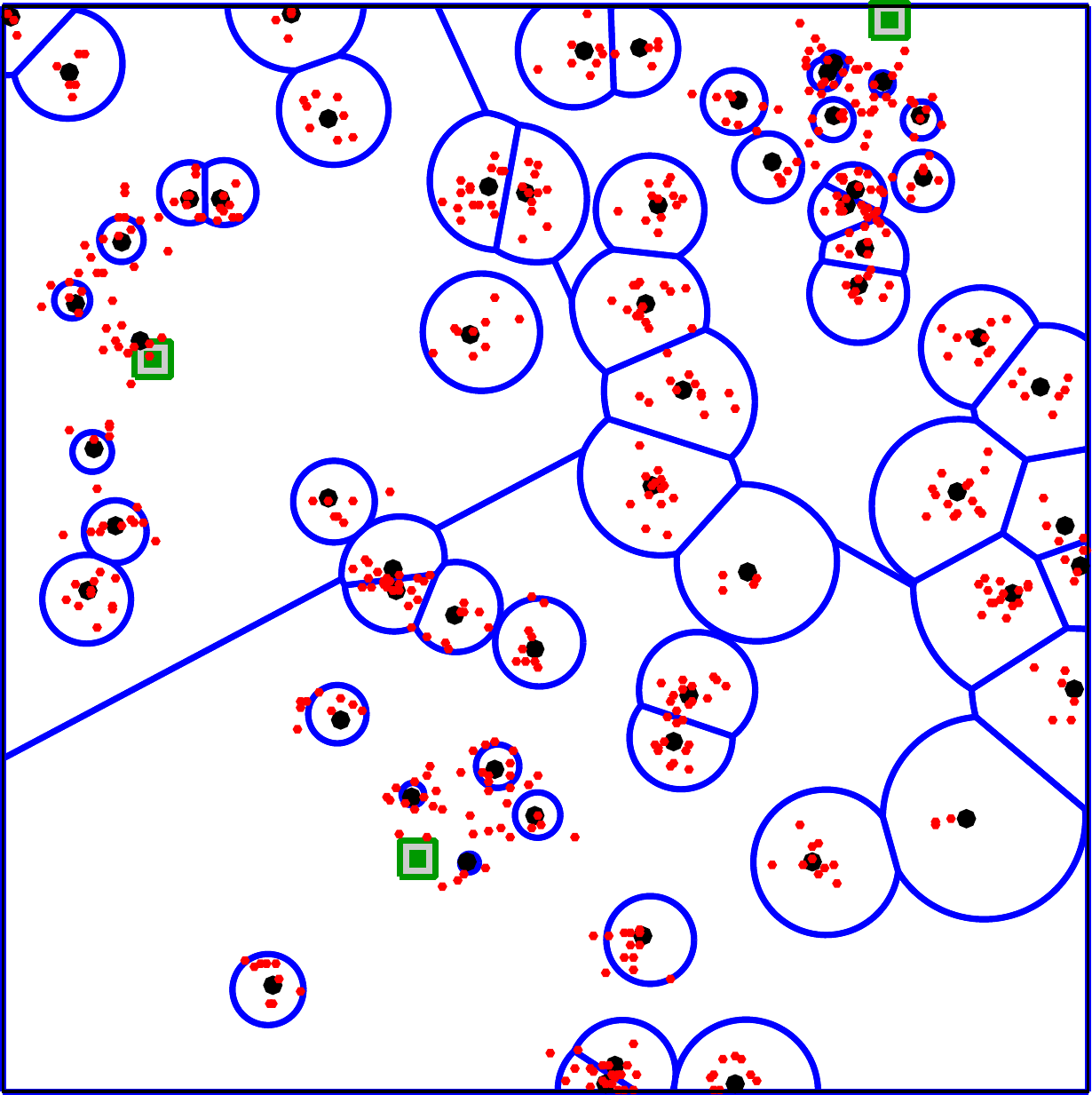}}\quad
  \subfloat[{\em Mixed} (clustered and uniformly distributed) user distribution (this paper).\label{fig::deployment_mixed}]{\includegraphics[width=0.3\textwidth]{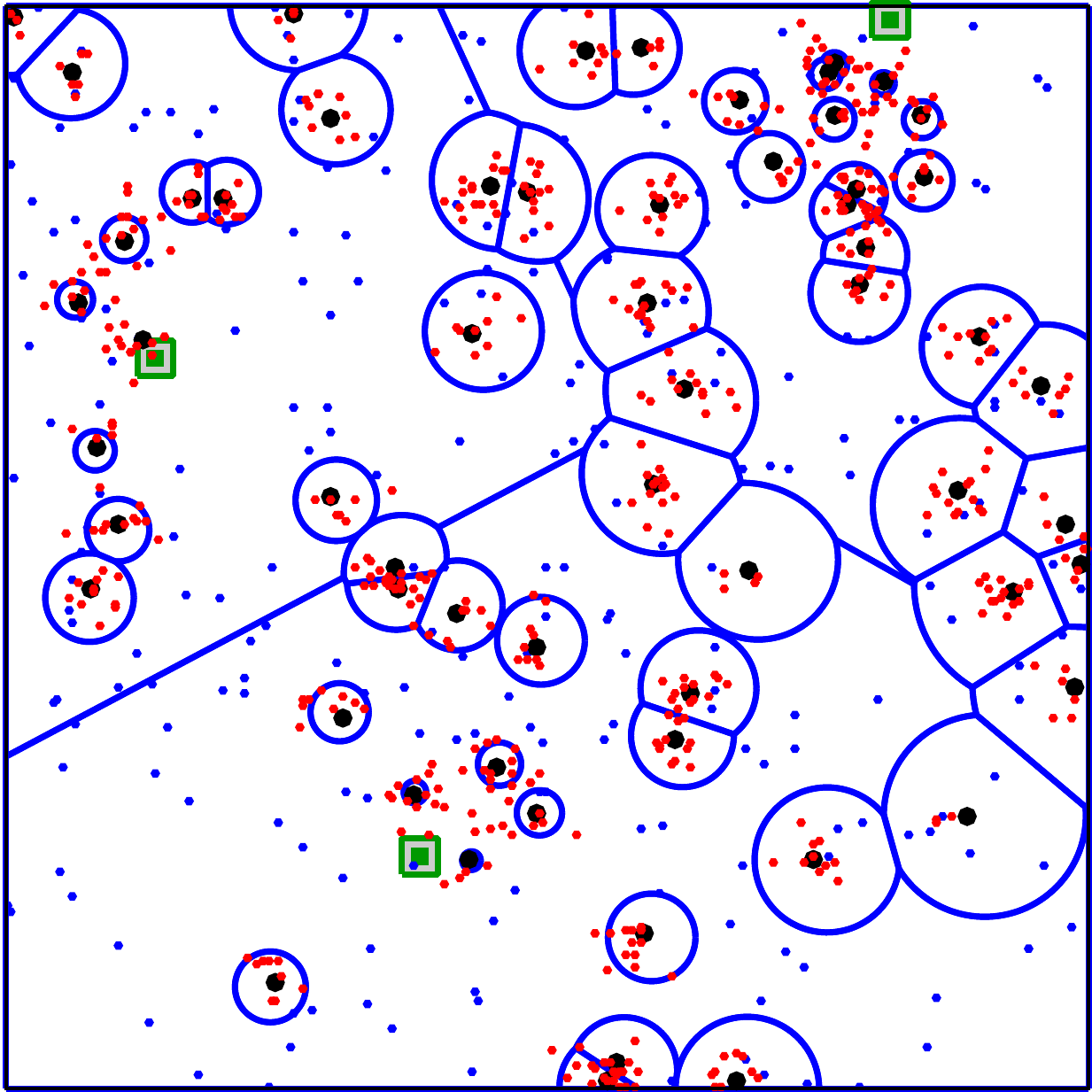}}
  \caption{{\small Macro (green squares) and small cell BSs (black  dots) are distributed as independent PPPs $\lambda_2=\lambda_2'=\lambda_1/10$. The uniformly distributed users are represented by small blue dots and clustered users by small red dots. The average number of users per cluster (wherever applicable) is $10$.}}
  \label{fig::deployment}
\end{figure*}
Although our analysis primarily focuses on users clustered around BS locations, we also consider users that are independently and homogeneously located over the network modeled as a PPP and use previously derived results for coverage \cite{xia}
 in conjunction to evaluate the overall coverage probability for any randomly chosen user in our HetNet setup  with mixed user distribution.
\subsubsection{System Design Insights} Our analysis leads to several system-level design insights.   
First, it can be observed that the coverage probability under the assumption of BS-user correlation is significantly greater than that derived under the assumption of independence.  While the assumption of independence of BS and user locations does simplify analyses, the resulting coverage probability predictions may be significantly pessimistic. That being said, our results concretely demonstrate that  the difference between the coverage probabilities corresponding to user-centric and independent BS deployment  becomes less significant as the cluster sizes (of user cluster) increase. In the limit of cluster size going to infinity, the new coverage results are shown to mathematically converge to the results obtained under independent user distribution assumption.
Second, as opposed to the previous works, the coverage probability of users clustered around BSs under interference-limited open access network is a function of BS transmission power. Our analysis shows that coverage probability can be improved by  increasing transmission power of small cell BSs located at centers of the user clusters.    
 
\section{System Model}\label{sec::system_model}
\subsection{BS Deployment}
Consider a $K$-tier HetNet, where BSs across tiers (or classes) differ in terms of their transmit powers and deployment densities. For mathematical convenience and notational simplicity, define $\ncalK=\{1,2,\dots,K\}$ as the indices of the $K$ tiers. The locations of the $k^{th}$-tier BSs are modelled by an independent homogeneous PPP $\Phi_k^{\rm (BS)}$ of density $\lambda_k^{\rm (BS)}>0$ ($k\in\ncalK$). The $k^{th}$-tier BSs are assumed to transmit at the same power $P_k$. 
 As is usually the case, we assume that a fraction of $k^{th}$-tier BSs are in open access for the user of interest while the rest are in closed access. The $k^{th}$-tier open and closed access BSs are modelled by two independent PPPs $\Phi_k^{\rm(BS,o)}$ and $\Phi_k^{\rm(BS,c)}$ with densities $\lambda_k$ and $\lambda_k'$, respectively, where $\Phi_k^{\rm (BS)} =\Phi_k^{\rm(BS,o)}\cup \Phi_k^{\rm(BS,c)}$ and $\lambda_k^{\rm(BS)}= \lambda_k + \lambda_k'$.

\subsection{User Distribution}\label{subsec::user_distribution}
Unlike prior art that focused almost entirely on the performance analysis of users that are uniformly distributed in the network independent of the BS locations, we focus on a correlated setup where users are more likely to lie closer to the BSs. Since small cells are usually deployed in the areas of high user density, this is a much more accurate approach for modeling HetNets compared to the one where users and BSs are both modeled as independent PPPs. We model this scenario by modeling the locations of the users by a PCP with one small cell deployed at the center of each user cluster. To maintain generality, we assume that a subset $\ncalB \subset \ncalK$ tiers out of  $K$ tiers  have clusters of users around the BSs. In particular, given the location of {a BS} in the $i^{th}$ tier acting as cluster center ($i\in\ncalB$), the users of the cluster are assumed to be symmetrically, independently, and identically distributed (i.i.d.) around it.   Union of all such locations of users around the  BSs of the $i^{th}$ tier forms a {PCP}~\cite{Stoyan_stochastic_1996, ganti}, denoted by $\PPP{i}^u$, where the  parent point process of $\PPP{i}^u$ is $\Phi_i^{\rm(BS)}$. 
  To maintain generality, we assume that the user location $\user\in \mathbb{R}^2$ with respect to its cluster center follows some arbitrary distribution with probability density function (PDF) $f_{\user}(\cdot)$, which may not necessarily be the same across tiers. This allows to capture the fact the cluster size may affect the choice of small cell to be deployed there. For instance, it may be sufficient to deploy a low power femtocell to serve a small cluster of users in a coffee shop, whereas a relatively higher power picocell may be needed to serve a cluster of users at a big shopping mall or at an airport. After deriving all the results in terms of the general distributions, we will specialize them to two cases of interest where $\Phi_i^{u}$ is modeled as: (i) {\em Thomas cluster process} in which the users are scattered according to a symmetric normal distribution of variance $\sigma^2_i$ around the BSs of {$\Phi_i^{\rm(BS)}$} \cite{HaeB2013}, hence,
\begin{equation*} 
f_{\user}(\mathbf{z})=\frac{1}{2\pi\sigma^{2}_i}\exp\left(-\frac{\|\mathbf{z}\|^2}{2\sigma^2_i}\right),\ \  \mathbf{z} \in \mathbb{R}^2,
\end{equation*} and (ii) {\em \matern\ cluster process} which assumes symmetric uniform spatial distribution of users around the cluster center within a circular disc of radius $\Ri$, thus 
\begin{align*}
f_{\user}(\mathbf{z})&=\begin{cases}\frac{1}{\pi\Ri^2} &\text{if } \|\mathbf{z}\|\leq \Ri\\
0&\text{otherwise}
\end{cases},
\end{align*}where $\mathbf{z}$ is a realization of the random vector $\user$. While our primary interest is in these clustered users, we also consider users that are homogeneously distributed over the network independent of the BS locations, for instance, pedestrians and users in transit. These users are better modeled by a PPP as done in literature (see \cite{dhillonHetNet,Mukherjee_Hetnet, xia,madhusudhanan_downlink_2014,di_renzo_average_2013} for a small subset).  Thus, in addition to the user clusters modeling users in the hotspots, we also consider an independent point process of users $\Phi^{u\rm(PPP)}$ which is a PPP of density $\lambda^{\rm(PPP)}$. {\figref{fig::deployment} shows the two-tier HetNet setup with high power macro-BSs overlaid with an independent PPP of denser but low power small cell BSs. \figref{fig::deployment_uniform} illustrates the popular system model used in the literature where users are modeled as $\Phi^{u\rm(PPP)}$ \cite{dhillon_tractable_2011, dhillonHetNet,xia, Mukherjee_Hetnet, madhusudhanan_downlink_2014}.  \figref{fig::deployment_thomas} highlights the correlated setup where users are only clustered around small cell BSs. The general scenario i.e. the {\em mixed} user distribution formed by the  
 superposition of PPP and PCP has been  depicted in  \figref{fig::deployment_mixed}.}


Since the downlink analysis at the location of a typical user of $\Phi^{u\rm (PPP)}$ is well-known, in this paper we will focus exclusively on the downlink performance of  a {typical user} of $\Phi_i^{u}$, which is a randomly chosen user from a randomly chosen cluster of  $\Phi_i^{u}$, also termed as the  {\em representative cluster}. In other words, we will primarily focus on the scenario depicted in \figref{fig::deployment_thomas} (and then extend our results and insights to scenario depicted in \figref{fig::deployment_mixed}). 
Since the PPPs are stationary, we can transform the origin to the location of this typical user. Quite reasonably, we assume that the BS at the center of the representative cluster is in open access mode. This assumption can be easily relaxed without much effort.  Denote the location of the representative cluster center by ${\bf y}_0 \in \Phi_i^{\rm(BS,o)}$. Now $\Phi_i^{\rm(BS,o)}$ can be partitioned into two sets: (i) representative cluster center ${\bf y}_0$, and (ii) the rest of the points $\Phi_i^{\rm(BS,o)} \setminus {\bf y}_0$. By Slivnyak's theorem, it can be argued that $\Phi_i^{\rm(BS,o)} \setminus {\bf y}_0$ has the same distribution as $\Phi_i^{\rm(BS,o)}$ \cite{HaeB2013}. For notational simplicity, we form an additional tier (call it tier $0$) consisting of a single point ${\bf y}_0$, i.e., $\Phi_0^{\rm(BS)} \equiv \Phi_0^{\rm(BS,o)} \equiv \{{\bf y}_0\}$. Then, the set of indices of all tiers is enriched to $\ncalK_1=\{0\}\cup \ncalK=\{0,1,2,\dots,K\}$.  The user can either connect to its own cluster center i.e. the BS in $\Phi_0^{\rm(BS,o)}$, or  to some other BS belonging to one of the tiers $\Phi_1^{\rm(BS,o)},\dots,\Phi_K^{\rm(BS,o)}$. {It will be evident in sequel that this construction will allow us to handle the link from the typical user to its cluster center separately.}
\subsection{Channel Model and User Association}\label{subsec::chmodel}
 The received power at the location of the typical user at origin from a BS at $\mathbf{y}_k\in\Phi_k^{\rm(BS)}$ ($k\in\ncalK_1$) {is modelled as} $P(\mathbf{y}_k)=P_kh_k\ncalV_k\|\mathbf{y}_k\|^{-\alpha}$, where, {$\alpha>2$} is the path loss exponent, $\h{y}{k}$ is the small-scale fading gain and $\ncalV_k$ is the shadowing gain. Under Rayleigh fading assumption, $\{h_k\}$ is a sequence of  i.i.d. exponential random variables (RVs) with $\h{x}{k}\sim \exp(1)$. For large scale shadowing, we assume $\{\ncalV_k\}$ to be sequence of i.i.d. log-normal RVs , i.e.,  $10\log\ncalV_k \sim \ncalN(\mu_k, \eta_k^2)$, with $\mu_k$ and $\eta_k$ respectively being the mean and standard deviation (in dB) of the channel power under shadowing. 
In this model, we assume average received power based cell selection in which a typical user connects to the BS that provides maximum received power averaged over small-scale fading. 
\begin{table}
\centering{
\caption{Summary of Notations}
\scalebox{0.8}{
\label{tab::notation}
\begin{tabular}{clc}

  \hline
   \hline

\textbf{Notation }                              & \textbf{Description                                                                                                                  } \\
\hline \hline 
$\Phi_k^{\rm(BS,o)},\lambda_k$ & PPP of  BSs of $k^{th}$ open access tier, density of $\Phi_k^{\rm(BS,o)}$ \\
$\Phi_k^{\rm(BS,c)},\lambda_k'$ & PPP of  BSs of $k^{th}$ closed access tier, density of $\Phi_k^{\rm(BS,c)}$ \\
$\ncalB$                              &Set of BS tiers that have users clustered around them\\  
$\Phi_i^{u}$                          & Point process modeling users clustered around BSs of $\Phi_i^{\rm(BS)}$ \\
$\Phi^{u\rm(PPP)}$                    & Locations of uniformly distributed users modeled as a  PPP\\ 
$\mathbf{y}_0, Y_0$                   & Location of cluster center in Euclidean space, $Y_0=\|\mathbf{y}_0\|$\\
$\Phi_0^{\rm(BS,o)}$                  & Tier 0 containing only the cluster center\\
$\Phi_k,\lam_k$                       & Equivalent PPP of $\Phi_k^{\rm(BS,o)}$ to incorporate shadowing, density of $\Phi_k$  \\
$\Phi_k',\lamc_k$                     & Equivalent PPP of $\Phi_k^{\rm(BS,c)}$ to incorporate shadowing, density of $\Phi_k'$  \\
$P_k, h_k, \ncalV_k$                  & Transmit power, small scale fading gain, shadowing gain \\
$\mathbf{y}_k$        & Actual location of a BS in $\Phi_k^{\rm(BS)}$   \\
$\mathbf{x}_k$  & Location of BS in transformed space ($\mathbf{x_k}=\mathcal{V}_k^{-\frac{1}{\alpha}}\mathbf{y}_k$)\\
$N_i$                                 & Average number of users per cluster  of $\Phi^u_{i}$            \\                                                                                   
$R_k$                                 & Modified distance of nearest BS $\in \Phi_k$, $R_k=\min \|{\bf x}_k\|$     \\
$b(\mathbf{0},r)$                  & {Disc with radius $r$ centered at origin}\\
$\ncalI_{o(j,k)}$                       & Interference from all BSs $\in\Phi_k$ when user connects to a BS $\in\Phi_j$\\
$\ncalI_{c(k)}$                         & Interference from all BSs $\in\Phi_k'$ \\
$\pc^{(i)}, \pc^{\rm(PPP)}$             & Coverage probability of a typical user in $\Phi_i^{u}$, $\Phi^{u\rm(PPP)}$  \\
$\pc$                                      &Overall coverage probability \\
\hline
\end{tabular}
}}
\end{table} 
The serving BS will be one from the $K+1$ candidate BSs from each tier. The location of such candidate serving BS from $\Phi_k^{\rm(BS,o)}$ can be denoted as: 
\begin{align*}
\mathbf{y}^*_k=\arg\max_{\Scale[0.85]{\substack{\mathbf{y_k}\in\\\Phi_k^{\rm(BS,o)}}}}P_k\ncalV_k\|\mathbf{y}_k\|^{-\alpha}=\arg\max\limits_{\substack{\mathbf{y_k}\in\\\Phi_k^{\rm(BS,o)}}}\Scale[0.85]{P_k\big(\ncalV_k^{-\frac{1}{\alpha}}\|\mathbf{y}_k\|\big)^{-\alpha}}.
\end{align*}
 Since the $0^{th}$ tier consists of only a single BS, i.e., the cluster center, there is {only one} choice of the candidate serving BS from $\Phi_0^{\rm(BS,o)}$, i.e., $\mathbf{y}_0^*\equiv \mathbf{y}_0$. The serving BS will be one of these candidate serving BSs, denoted by
 \begin{align*}
 {\mathbf{y}^*}=\arg\max\limits_{\bf{y}\in\{\bf{y_k^*}\}}P_k\ncalV_k\|\mathbf{y}\|^{-\alpha}.
\end{align*}    
{Using the displacement theorem of PPPs \cite[Section 1.3.3]{baccelli2009stochastic} , it was shown in \cite{DhiAndJ2014,keeler_sinr-based_2013} that if each point in a PPP $\Phi_k^{(\rm BS,o)}$ ($\Phi_k^{(\rm BS,c)}$) is independently displaced such that the transformed location becomes $\mathbf{x}_k=\ncalV_k^{-\frac{1}{\alpha}}\mathbf{y}_k$, then, the resultant point process remains a PPP, which we denote by $\Phi_k$ ($\Phi_k'$)   with density $\lam_k=\lambda_k\nbbE\left[\ncalV_k^{\frac{2}{\alpha}}\right]$ ($\lamc_k=\lambda_k'\nbbE\left[\ncalV_k^{\frac{2}{\alpha}}\right]$).  This transformation is valid for  any arbitrary distribution of $\ncalV_k$ with PDF $f_{\ncalV_k}(\cdot)$ as long as $\nbbE(\ncalV_k^{\frac{2}{\alpha}})$ is finite, which is indeed true for log-normal distribution.     Consequently, we can express instantaneous  received power from a BS~$\in\Phi_k$ as $P_kh_k\|\mathbf{x}_{k}\|^{-\alpha}$. Then, the location of candidate serving BS in $\Phi_k$  can be written as} 
\begin{align*} 
  \mathbf{x}_k^*=\arg\max\limits_{\mathbf{x_k}\in\Phi_k}P_k\|\mathbf{x}_k\|^{-\alpha}.
 \end{align*}
For $k=0$,  we apply  similar transformation to the point  ${\bf y}_0\in\Phi_0^{\rm(BS,o)}$ and denote the transformed process as $\Phi_0$ where $\mathbf{x}_0\equiv\ncalV_0^{-\frac{1}{\alpha}}\mathbf{y}_0$. Then, the serving BS at $\bf{x}^{*}$ will be
 \begin{align*}
 {\mathbf{x}^*}=\arg\max\limits_{\bf{x}\in\{\bf{x_k^*}\}}P_k\|\mathbf{x}\|^{-\alpha}.
 \end{align*}  
It is worth noting that in the absence of shadowing, the candidate serving BS from a given tier will be the BS closest to the typical user from that tier in terms of the Euclidean distance. This is clearly not true in the presence of shadowing because of the possibility of a farther off BS providing higher average received power than the closest BS. However, by applying displacement theorem, the effect of shadowing gains has been incorporated at the modified locations ${\bf x}_k$ such that the strongest BS in the equivalent PPP $\Phi_k$ is also the closest in terms of the Euclidean distance. As demonstrated in the literature (e.g., see\cite{ keeler_sinr-based_2013,DhiAndJ2014}) and the next two Sections, this simplifies the coverage probability analysis in the presence of shadowing significantly.  
For notational simplicity, let us define the association event to tier $j$ as $S_{\Phi_j}$ such that $\mathbf{1}_{S_{\Phi_j}}=\mathbf{1}(\mathbf{x}^*=\mathbf{x}^*_j)$ (here $\mathbf{1}(\cdot)$ is  the indicator function). The  Signal-to-Interference and Noise Ratio ($\sinr$) experienced by a typical user at origin when $\mathbf{1}_{S_{\Phi_j}}=1$ can be expressed as:
   \begin{align}
   \sinr(\|\mathbf{x}^*\|)\equiv\frac{P_j h_j  \|\mathbf{x}^*\|^{-\alpha}}{N_0+\sum_{k\in\ncalK_1}\sum_{{\mathbf{x}_k\in\Phi_k\cup\Phi_k'\setminus \{\mathbf{x}^*\}}}P_kh_k\|\mathbf{x}_k\|^{-\alpha}},
   \label{eq::sir_def}
\end{align} where $N_0$ is the thermal noise power.   For quick reference, the notations used in this paper are summarized in  Table~\ref{tab::notation}.
\begin{remark}
While we transform all the PPPs to equivalent PPPs to incorporate shadowing,  
the impact of shadowing on the link between the typical user and  its cluster center, i.e., ${\bf x}_0=\ncalV_0^{\frac{1}{\alpha}}{\bf y}_0$,
needs to be handled separately. For this, we have   two alternatives. First is to find the distribution of ${\bf x}_0$ as a function of the distributions  of $\ncalV_0$ and ${\bf y}_0$. Second is to proceed with the analysis by conditioning on shadowing variable $\ncalV_0$ and decondtioning at the last step. We take the second approach since it gives simpler intermediate results which can be readily used for no-shadowing scenario by putting $\ncalV_k\equiv 1$. 
\label{rem::rem1}
\end{remark}
\section{Association Probability and Serving Distance}\label{sec::association}
This is the first technical section of the paper, where we derive the probability that a typical user is served by a given tier $j \in \ncalK_1$, which is usually termed as the {\em association probability}. We will then derive the distribution of $\|\mathbf{x}^*\|$ conditioned on $S_{\Phi_j}$, i.e., the distance from the typical user to its serving BS conditioned on the  the event that it belongs to the $j^{th}$ open access tier.  
Recall that the candidate serving BS located at ${\bf x}_k$ from the equivalent PPP $\Phi_k$ is the one that is nearest to the typical user located at the origin.
Let us {call} $R_k=\|\mathbf{x}^*_k\|$ as the RV denoting the  distance from the typical user  to the nearest point of $\Phi_k$.  Since $\Phi_k$ ($k\in\ncalK$) are independent homogeneous PPPs, the distribution of $R_k$,  $k\in\ncalK$, is~\cite{HaeB2013}
\begin{subequations}
\begin{align}
&\text{PDF:}\ &f_{R_k}(r_k)&=2\pi\lam_k r_k\exp(-\pi\lam_k r_k^2) & r_k \geq 0, \label{eq::nn_distribution} \\
&\text{CCDF:}\ &\overline{F}_{R_k}(r_k)&=\exp(-\pi\lam_k r_k^2) & r_k \geq 0. \label{eq::nn_distribution_ccdf}
\end{align}
\label{eq::nn}
\end{subequations}
In a similar way, we can define modified distance $R_0=\|{\bf x}_0\|=\ncalV_0^{-\frac{1}{\alpha}}\|{\bf y}_0\|$. As noted in Remark \ref{rem::rem1}, we will proceed with the analysis by conditioning on the shadowing gain $\ncalV_0$ and then deconditioning on $\ncalV_0$ at the very end. Since $R_0$ is just a scaled version of $\|{\bf y}_0\|$, it suffices to find the distribution of $Y_0 \equiv \|{\bf y}_0\|$, which we do next.

 Recall that the typical user is located at the origin, which means the relative location of the cluster center with respect to the typical user, i.e., ${\bf y}_0$, has the same distribution as that of $\user$. Using standard transformation technique from Cartesian to polar coordinates, we can obtain the distribution of distance $Y_0$ from the joint distribution of position coordinates $(t_1, t_2)$, where ${\mathbf{y}}_0=(t_1,t_2)$ is in Cartesian domain. Let us denote the joint PDF of the polar coordinates $(Y_0,\Theta)$ as  $f_{Y_0,\Theta}(\cdot)$. Then  
\begin{align}
f_{Y_0,\Theta}(y_0,\theta)=f_{{\mathbf{y}}_0}(t_{1},t_{2})\times \left| \partial \left(\dfrac{t_1,t_2}{y_0,\theta}\right)\right|,
\end{align}where
\begin{align*}
\partial \left(\dfrac{t_1,t_2}{y_0,\theta}\right) = \begin{bmatrix}
  \dfrac{\partial t_{1}}{\partial y_0} & \dfrac{\partial t_{1}}{\partial\theta}\\[1em]
  \dfrac{\partial t_{2}}{\partial y_0} & \dfrac{\partial t_{2}}{\partial\theta} \end{bmatrix}.
\end{align*}
From the joint distribution, the marginal distribution of distance $Y_0$ can now be computed by integrating over $\theta$ as%
\begin{align*}
f_{Y_{0}}(y_0)&=\int_0^{2\pi}f_{Y_0,\Theta}(y_0,\theta) {\rm d}\theta.
\end{align*}
\begin{remark}
In the special case when $\Phi_i^{u}$ is a Thomas cluster process, user coordinates in Cartesian domain are i.i.d. normal RVs with variance $\sigma_i^2$.  Then, $Y_0$ is Rayleigh distributed with PDF and CCDF~\cite{AfsDhiJ2015}:
\begin{subequations}
\begin{align}
&\text{PDF:}\ &f_{Y_0}(y_0)&=\frac{y_0}{\sigma_i^2}\exp\left(\frac{-y_0^2}{2\sigma_i^2}\right),&\ y_0\geq 0, \label{eq::dist_thomas}\\
&\text{CCDF:}\ &\overline{F}_{Y_0}(y_0)&=\exp\left(\frac{-y_0^2}{2\sigma_i^2}\right),&\ y_0\geq 0. \label{eq::dist_thomas_ccdf}
\end{align}
\label{eq::thomas}
\end{subequations}
\end{remark}
\begin{remark}
If $\Phi_i^{u}$ is a \matern\ cluster process, the PDF and CCDF of $Y_0$ are:
\begin{subequations}
\begin{align}
&\text{PDF:}\ &f_{Y_0}(y_0)&=\frac{2y_0}{\Ri^2},& 0\leq y_0\leq\Ri, \label{eq::dist_matern}\\
&\text{CCDF:}\ &\overline{F}_{Y_0}(y_0)&=\frac{\Ri^2-y_0^2}{\Ri^2}, & 0\leq y_0\leq\Ri. \label{eq::dist_matern_ccdf}
\end{align}
\label{eq::matern}
\end{subequations}
\end{remark}
\subsection{Association Probability}
To derive association probability, {let us first  characterize the association event $S_{\Phi_j}$ as}: $\mathbf{1}_{S_{\Phi_j}}=$
\begin{equation}
\mathbf{1}(\arg\max\limits_{k\in\ncalK_1}P_k R_k^{-\alpha}=j)
=\bigcap\limits_{{k\in\ncalK_1}}\mathbf{1}\left(R_k>\pow{j}{k}R_j\right),
\label{eq::def_s_phi}
\end{equation}
where $\pow{j}{k}=\left(\frac{P_k }{P_j }\right)^{\scriptscriptstyle 1/\alpha}$ and $\mathbf{1}(\cdot)$ is the indicator function of the random vector $\mathbf{R}=[R_0,R_1,...,R_k]$. Note that since the $0^{th}$ tier is derived from the $i^{th}$ tier, we have  $P_0 \equiv P_i$. The association probability for each tier is now defined as follows.
\begin{ndef}{\text{Association probability,}} $\ncalA_j$ for $j^{th}$ tier, $\forall j \in \ncalK_1$ is defined as  the probability that the typical user will be served by the $j^{th}$ tier. It can be mathematically expressed as
\begin{equation}
\ncalA_j=\nbbP(S_{\Phi_j}).
\label{eq::association_prob_defn}
\end{equation}
\end{ndef}
 The following Lemma deals with the conditional association probability to $\Phi_j$.
\begin{lemma}
\label{lemma::association_prob_sh}
 Conditional association probability of the $j^{th}$ tier given $\ncalV_0=v_0$ is
\begin{align}
\ncalA_{j|v_0}=\begin{cases}
\nbbE_{Y_0}\left[\prod\limits_{k=1}^K
\overline{F}_{R_k}(\pow{0}{k}v_0^{-\frac{1}{\alpha}}Y_0)\right]& \text{if }j=0;\\
\nbbE_{\mod_j}\bigg[\overline{F}_{Y_0}(v_0^{\frac{1}{\alpha}}\pow{j}{0}R_j)\prod\limits_{\substack{k=1\\k\neq j}}^K\overline{F}_{\mod_k}(\pow{j}{k}R_j)\bigg] &\text{if }j\in\ncalK.
\end{cases}
\label{eq::association_prob_main_lemma}
\end{align}    
\end{lemma}
\begin{proof}
See Appendix~\ref{app::association_prob_sh}.
\end{proof}
\begin{remark}
Association probabilities of the $j^{th}$ tier can be obtained by taking expectation over $\condA{j}$ with respect to $\ncalV_0$, i.e.,
\begin{align*}
\ncalA_j=\nbbE_{\ncalV_0}(\condA{j}).
\end{align*}
\end{remark}
From   Lemma \ref{lemma::association_prob_sh}, we can obtain the expressions for the association probabilities to different open access tiers when $\Phi_i^{u}$ is Thomas or \matern\ cluster process. The conditional probabilities in these cases can be reduced to closed form expressions. The results are presented next.
\begin{cor}\label{corr:thomas_association}
When $\Phi_i^{u}$ is a Thomas cluster process, conditional  association probability of the $j^{th}$ tier given $\ncalV_0=v_0$ is: 
\begin{align}
  &\ncalA_{j|v_0}=\frac{\lam_j}{\sum\limits_{k=0}^K \pow{j}{k}^2\lam_k}, \ \forall j\in\ncalK_1,
  \label{eq::association_thomas_corr_1}
  \end{align}
where $\lam_0$ is defined as
$ \lam_0=\frac{v_0^{\frac{2}{\alpha}}}{2\pi\sigma_i^2}.$
\label{cor::association_thomas}
\end{cor}
\begin{proof}
See Appendix~\ref{app::association_thomas}.
\end{proof}
\begin{cor}\label{cor::association_matern}
If $\Phi_i^{u}$ is {a} \matern\ cluster process, conditional  association probability of the $j^{th}$ tier given $\ncalV_0=v_0$ is: $\ncalA_{j|v_0}=$
\begin{align}\label{eq::association::matern}
\begin{cases}
 \frac{v_0^{\frac{2}{\alpha}}}{\Ri^2\ncalZ_0}\bigg(1-\exp\bigg(-{v_0^{-\frac{2}{\alpha}}}\ncalZ_0\Ri^2\bigg)\bigg) &\text{if $j=0$}\\
\frac{\pi\lam_j}{{\cal Z}_j}-\frac{\lam_j\pi\pow{j}{0}^2v_0^{\frac{2}{\alpha}}}{\Ri^2{\cal Z}_j^2}\bigg(1-\exp\big(-\frac{{\cal Z}_j\Ri^2}{\pow{j}{0}^2v_0^{\frac{2}{\alpha}}}\big)\bigg)
 &\text{if } j\in\ncalK
\end{cases},
\end{align}
where $\ncalZ_j=\pi\sum\limits_{k=1}^K\lam_k\pow{j}{k}^2,$ $\forall\:j\in\ncalK_1$.
\end{cor}
\begin{proof}
See Appendix~\ref{app::association_matern}.
\end{proof}
\subsection{Serving Distance Distribution}
In this section, we derive the distribution of $\|\mathbf{x}^*\|$ when $\mathbf{1}_{S_{\Phi_j}}=1$ , i.e., the serving  distance from the typical user to its serving BS when it is in $\Phi_j$. We will call this RV $\serv_j$. Conditioned on $S_{\Phi_j}$,  $\serv_j$ is simply the distance to the nearest BS in $\Phi_j$. Hence $\serv_j$ is related to $R_j$ as $\serv_j=R_j|S_{\Phi_j}$. The conditional PDF of $\serv_j$ given $\ncalV_0=v_0$ is derived in the next Lemma.  
\begin{lemma}
\label{lemma::serving_distance_dist_sh}
Conditional distribution of serving distance $\serv_j$ at $\ncalV_0=v_0$ is obtained by $ \fservcond{j}(w_j|v_0)=$
 \begin{align} \label{eq::serving_distance_dist}
 \begin{cases}
 \frac{1}{\condAs{0}}\prod\limits_{k=1}^{K}v_0^{\frac{1}{\alpha}}\overline{F}_{\mod_k}\left(\pow{0}{k}w_0\right)f_{Y_0}(v_0^{\frac{1}{\alpha}}w_0),&\text{if }j=0,\\
 \frac{1}{\condAs{j}}\overline{F}_{Y_{0}}(v_0^{\frac{1}{\alpha}}\pow{j}{0}w_j)\prod\limits_{\substack{k=1\\k\neq j}}^K\overline{F}_{\mod_{k}}(\pow{j}{k}w_j)f_{\mod_j}(w_j),& \text{if }j\in\ncalK.
\end{cases}
 \end{align}
\end{lemma} 
\begin{proof}
 See Appendix~\ref{app::serving_distance_dist_sh}.
\end{proof}

Further, we obtain closed-form expressions of $\fservcond{j}(\cdot)$ for Thomas and \matern\   cluster processes by putting the corresponding PDFs and CCDFs in the following Corollaries.  
\begin{cor}\label{corr::thomas_serving}
If $\Phi_i^{u}$ is Thomas cluster process, conditional PDF of serving distance given $\ncalV_0=v_0$ can be expressed as
\begin{align}
   \fservcond{j}(w_j|v_0)&=
   \Scale[0.95]{\frac{2\pi\lam_j}{\ncalA_{j|v_0}}}\exp\bigg(-\Scale[0.9]{\pi\big(\sum\limits_{k=0}^{K}\pow{j}{k}^2\lam_k}\big)w_j^{2}\bigg)w_j,\forall j\in \ncalK_1.
\end{align}
\label{cor::f_dist_thomas}
\end{cor}
\begin{proof}
See Appendix~\ref{app::f_dist_thomas}.
\end{proof}
\begin{cor}
If $\Phi_i^{u}$ is \matern\ cluster process, the  conditional distribution of serving distance $\serv_j$ given $\ncalV_0=v_0$ can be expressed as $\fservcond{j}(w_j|v_0)=$
\begin{align*} 
&\begin{cases}
   \frac{1}{\ncalA_{0|v_0}}\exp\left(-\pi\sum\limits_{k=1}^{K}\lam_k\pow{0}{k}^2w_0^2\right)\frac{2v_0^{\frac{2}{\alpha}}w_0}{\Ri^2}&\text{if }j=0\\
\frac{1}{\ncalA_{j|v_0}}2\pi\lam_j\exp\left(-\pi\sum\limits_{k=1}^{K}\lam_k\pow{j}{k}^2w_j^{2}\right)
\frac{\Ri^2-v_0^{\frac{2}{\alpha}}w_j^2}{\Ri^2}w_j&\text{if }j\in\ncalK\\
\end{cases},
\end{align*}
where $0\leq w_j\leq \Ri$. For $w_j>\Ri$, $\fservcond{j}(w_j|v_0)=0,\ \forall j\in\ncalK_1$.
\end{cor}
\begin{proof}
Substituting $f_{Y_0}(\cdot)$ for \matern\ cluster process from Eq.~\ref{eq::dist_matern} and CCDF of $R_k$ from Eq.~\ref{eq::nn_distribution_ccdf} in Eq.~\ref{eq::serving_distance_dist} and proceeding as before, $\fservcond{0}(\cdot)$ can be derived. Similarly, $\fservcond{j}(\cdot)$ is obtained by substituting $\overline{F}_{Y_0}(\cdot)$ from Eq.~\ref{eq::dist_matern_ccdf}.  For $w_j>\Ri$, $\fservcond{j}(w_j|v_0)=0$, $\forall j\in\ncalK_1$ as $f_{Y_0}(\cdot)$ and $\overline{F}_{Y_0}(\cdot)$ take zero value beyond this range. 
\end{proof}
\section{Coverage Probability Analysis}\label{sec::coverage_prob}
This is the second technical section of the paper where we use the association probability and the distance distribution results obtained in the previous section to derive easy-to-use expressions for the coverage probability of a typical user of $\Phi_i^{u}$ in a {\em user-centric} deployment. 

 According to the association policy, it is easy to deduce that if the typical user is served by a BS $\in \Phi_j$ located at a distance $W_j$,  there exist no $k^{th}$ tier BSs, $\forall k \in \ncalK_1$, within a disc of radius $\pow{j}{k}W_j$ centered at the location of typical user (origin). We denote this \textit{exclusion disc}  by $b(\mathbf{0},\pow{j}{k}W_j)$. Assuming association with {the}  $j^{th}$ tier, the total interference experienced by the typical user originates from two independent sets of BSs: (i) $\cup_{k\in \ncalK_1} \Phi_k \setminus b(\mathbf{0},\serv_j)$, the set of open access BSs lying beyond the exclusion zone $b(\mathbf{0},\serv_j)$ and (ii) $\cup_{k\in \ncalK} \Phi'_k$, the set of closed access BSs. As all the interferers from the $k^{th}$ open access tier will lie outside $b(\mathbf{0},\pow{j}{k}\serv_j)$, we define interference from the $k^{th}$ open-access tier as $\ncalI_{o(j,k)}(\serv_j)= \sum_{\mathbf{x}_k\in \Phi_k \setminus b(\mathbf{0},\pow{j}{k}\serv_j)}P_k \h{x}{k} \|\mathbf{x}_k\|^{-\alpha}$.
We express the total contribution of interference from all open access tiers as 
\begin{align*}
\ncalI_{o(j)}(\serv_j)=\sum_{k =0}^{K}\ncalI_{o(j,k)}(\serv_j). 
\end{align*}
It is clear that the interference from the open-access tiers defined above depends on the serving distance $\serv_j$. However, it is not the case with the closed access tiers. Recall that since the closed access tiers do not participate in the cell selection procedure, there is no exclusion zone in their interference field. In particular, the closed access BSs may lie closer to the typical user than its serving BS. We denote the closed access interference by $\ncalI_c=\sum_{k=1}^K\ncalI_{c(k)}$, where $\ncalI_{c(k)}$ is the interference from all the BSs of the $k^{th}$ closed access tier $\Phi_k'$. Using the variables defined above, we can now express $\sir$ defined in Eq.~\ref{eq::sir_def} at the typical user when it is served by the BS located at a distance $\serv_j$  in a compact form as a function of the RV $\serv_j$ as: $
\sir(\serv_j)=
\frac{P_j\h{x}{j} \serv_j^{-\alpha}}{ {\ncalI_{o(j)}(\serv_j)}
+\ncalI_{c}}.$

\subsection{Coverage Probability}  \label{subsec::coverage_prob}
{A typical user is said to be in coverage if} $\sir(\serv_j)>\tau$, where $\tau$ denotes modulation-coding specific $\sir$ threshold required for successful reception. The coverage probability can now be formally defined as follows. 
\begin{ndef}[Coverage probability] Per-tier coverage probability for $\Phi_j$ can be defined as the probability that the typical user of $\Phi_i^{u}$ is in coverage conditioned on the fact that it is served by a  BS from $\Phi_j$. Mathematically,
\begin{align}
   \pc_j^{(i)}&=\nbbE({\bf 1}(\sir(\serv_j)>\tau)).
           \label{eq::coverage_per_tier_no_sh}
\end{align}
The total coverage probability $\pc^{(i)}$ can now be defined in terms of the per-tier coverage probability as  
\begin{align}
   \pc^{(i)}&=\sum\limits_{j=0}^{K}\ncalA_j\pc^{(i)}_j,
           \label{eq::coverage_total_no_sh}
\end{align}
where $\ncalA_j$ is given by Eq.~\ref{eq::association_prob_defn}.
\end{ndef}
%
With the expressions of $\condA{j}$ and $\fservcond{j}(\cdot)$ at hand, we focus on the derivation of  coverage probability $\pc^{(i)}$. Note that using the Rayleigh fading assumption along with the fact that the open access interference terms $\{\ncalI_{o(j,k)}\}$ and the closed access interference terms $\{\ncalI_{c(k)}\}$ are all independent of each other, we can express the per-tier coverage probability in terms of the product of Laplace transforms of these interference terms. This result was presented for a special case of Thomas cluster process in the conference version of this paper~\cite{Saha1605:Downlink} (for $K$-tier HetNets) as well as in~\cite{mankarmodeling} (for single-tier cellular networks). 
\begin{theorem}[Coverage probability]  \label{thm:coverage_theorem}
Conditional per-tier coverage probability of the typical user from $\Phi_i^{u}$ {given that} the serving BS being from the $j^{th}$ tier and $\ncalV_0=v_0$ is:   $\pc^{(i)}_{j|v_0}=$
\begin{align}
&\int\limits_{\scriptscriptstyle w_j>0}\ncalL_{\ncalI_{o(j,0)|\ncalV_0}}\left(\frac{\tau w_j^{\alpha}}{P_j}|v_0\right)\prod\limits_{k =1}^{K}\scalebox{0.95}{$\ncalL_{\ncalI_{o(j,k)}}\left(\frac{\tau w_j^{\alpha}}{P_i}\right)$}\scalebox{0.95}{$\ncalL_{\ncalI_{c(k)}}\left(\frac{\tau w_j^{\alpha}}{P_j}\right)$}\:\Scale[0.89]\notag\\&{\qquad\qquad   \times \exp\bigg(-\frac{\tau N_0 w_j^{\alpha}}{P_j}\bigg)\:\fservcond{j}(w_j|v_0)\:{\rm d}w_j},
 \label{eq::coverage_theorem_sh}
\end{align}
and the coverage probability of a typical user from $\Phi_i^u$ can be expressed as
\begin{equation}
\pc^{(i)}=\nbbE_{\ncalV_0}\left[\sum\limits_{j=0}^K\ncalA_{j|\ncalV_0}\pc^{(i)}_{j|\ncalV_0}\right],
\label{eq::coverage_theorem_coverage}
\end{equation}
where $\ncalL_{\ncalI_{o(j,0)|\ncalV_0}}\left(s|v_0\right)$ is the conditional Laplace transform of $\ncalI_{o(j,0)}$, i.e., $\ncalL_{\ncalI_{o(j,0)|\ncalV_0}}\left(s|v_0\right)\equiv\nbbE \left[\exp(-s \ncalI_{o(j,0)})|\ncalV_0=v_0\right]$ and $\ncalL_{\ncalI_{o(j,k)}}(s) \equiv \nbbE \left[\exp(-s \ncalI_{o(j,k)}) \right]$;  $\ncalL_{\ncalI_{c(k)}}(s) \equiv \nbbE \left[\exp(-s \ncalI_{c(k)}) \right]$ respectively denote the Laplace transforms of interference of  the $k^{th}$ open and closed access tiers ($k\in\ncalK$).
\label{th::coverage}
\end{theorem}

\begin{proof}
See Appendix~\ref{app::coverage}. 
\end{proof}
{Note that the conditioning on $\ncalV_0$  appears only in the first term, i.e. the Laplace transform of $\ncalI_{o(j,0)}$ since the interference from the BS at cluster center is  only influenced by $\ncalV_0$ while the other interference terms are independent of $\ncalV_0$.}
\subsection{Laplace Transform of Interference}
 As evident from Theorem~\ref{th::coverage}, the Laplace transform of interference from different tiers are the main components of  the  coverage probability expression. 
The following three Lemmas deal with the Laplace transforms of the interference from different tiers. We first focus on the interference originating from all the open access tiers except the interference from the $0^{th}$ tier (i.e. the BS at cluster center) which requires separate treatment. 
\begin{lemma}
Given a typical user of $\Phi_i^{u}$ is served by a BS $\in \Phi_j$ ($j\in \ncalK$) at a distance $W_j=w_j$,  Laplace transform of $\ncalI_{o(j,k)}$, $\forall k\in \ncalK$, evaluated at $s=\frac{\tau w_j^{\alpha}}{P_j}$ is
 \begin{align}
 \ncalL_{\ncalI_{o(j,k)}}\left(\frac{\tau w_j^{\alpha}}{P_j}\right)&=\exp\biggl(-\pi\pow{j}{k}^2\lam_k \G w_j^2\biggr),
\label{eq::I_o_jk}\\
\text{with}\quad {\G} &=\frac{2\tau }{\alpha-2} {}_2\ncalF_{1}\left[1,1-\frac{2}{\alpha};2-\frac{2}{\alpha},-\tau\right]\label{eq::G}
,\end{align}
where ${}_2\ncalF_{1}[a,b,c,t]=\frac{\Gamma(c)}{\Gamma(b)\Gamma(c-b)}\int\limits_{0}^1\frac{z^{b-1}(1-z)^{c-b-1}}{(1-tz)^{a}}{\rm d}z$ is  Gaussian Hypergeometric function.
\label{lemma::L_I_0(j,k)}
\end{lemma}
\begin{proof} 
The proof follows on the same lines as \cite[Theorem 1]{xia}. For completeness, the proof is provided in Appendix~\ref{app::I_o_jk}. 
 \end{proof}
After dealing with the interference from all open access tiers $\Phi_k$ $\forall k\in \ncalK$, we now focus on the $0^{th}$ tier, which consists of only the cluster center. 
\begin{lemma}
\label{lemma::I_o_j0}
Given a typical user of $\Phi_i^{u}$ connects to the BS $\in \Phi_j$ with $j\in\ncalK$ at a distance $W_j=w_j$, the Laplace transform of $\ncalI_{o(j,0)}$ at $s=\frac{\tau w_j^{\alpha}}{P_j}$ conditioned on $\ncalV_0=v_0$ is:  $\ncalL_{\ncalI_{o(j,0)|\ncalV_0}}\big(\frac{\tau w_j^{\alpha}}{P_j}\bigg|v_0\big)=$
\begin{align}
\label{eq::Io_j0_general}
\int\limits_{y_0>v_0^{\frac{1}{\alpha}}\pow{j}{i}w_j}\frac{1}{1+\tau\bigg(\frac{y_0}{v_0^{\frac{1}{\alpha}}\pow{j}{i}w_j}\bigg)^{-\alpha}}\frac{f_{Y_0}(y_{0})}{\overline{F}_{Y_0}(v_0^{\frac{1}{\alpha}}\pow{j}{i}w_j)}\:{\rm d}y_{0}.
\end{align}
\end{lemma}
\begin{proof}{Recall that since the $0^{th}$ tier is created by the BS at cluster center which actually belongs to  the $i^{th}$ open-access tier,  the transmit power $P_0\equiv P_i$.}  If the serving BS  $\in\Phi_j$ ($j\in\ncalK$) lies at a distance $W_j=w_j$, due to the formation of virtual exclusion zone around the typical user, the cluster center acting as an interferer will lie outside $b(\mathbf{0},\pow{j}{i}w_j)$. {Thus, the PDF of  distance from the typical user to cluster center conditioned on $Y_0>v_0^{\frac{1}{\alpha}}\pow{j}{i}w_j$ is $f_{Y_0}(y_{0}|Y_0>v_0^{\frac{1}{\alpha}}\pow{j}{i}w_j) = \frac{f_{Y_0}(y_{0})}{\overline{F}_{Y_0}\left(v_0^{\frac{1}{\alpha}}\pow{j}{i}w_j\right)}$, where $y_{0}> v_0^{\frac{1}{\alpha}}\pow{j}{i}w_j$.} The conditional Laplace transform $\ncalL_{\ncalI_{o(j,0)|\ncalV_0}}(s|v_0)$  can be expressed as: $\nbbE_{Y_0}\big(\nbbE_{h_0}\big(\exp\big(-sP_ih_0v_0Y_0^{-\alpha}\big)\big)|\mod_0>\pow{j}{i}w_j\big)$
 \begin{align} &\myeq{a}\nbbE_{Y_0}\left[\frac{1}{1+sP_iv_0Y_{0}^{-\alpha}}|Y_0>v_0^{\frac{1}{\alpha}}\pow{j}{i}w_{j}\right]\label{eq::interferece_cc_intermediate_step}\\
 &=\int\limits_{\scalebox{0.55}{$y_0>v_0^{\frac{1}{\alpha}}\pow{j}{i}w_{j}$}}\frac{1}{1+sP_iv_0y_0^{-\alpha}}f_{Y_0}(y_{0}|Y_0>v_0^{\frac{1}{\alpha}}\pow{j}{i}w_j)\:{\rm d}y_{0}\notag\\&=\int\limits_{\scalebox{0.55}{$y_0>v_0^{\frac{1}{\alpha}}\pow{j}{i}w_{j}$}}\frac{1}{1+sP_iv_0y_{0}^{-\alpha}}\frac{f_{Y_0}(y_{0})}{\overline{F}_{Y_0}\left(v_0^{\frac{1}{\alpha}}\pow{j}{i}w_j\right)}\:{\rm d}y_{0},\notag
\end{align}
where (a) follows from $h_0\sim\exp(1)$. This completes the proof. 
\end{proof}
In the next Corollary, we provide closed form upper and lower bounds on the Laplace transform of interference from the BS at $0^{th}$ tier.  The lower bound is obtained by placing the BS located at the cluster-center of the typical user on the boundary of the exclusion disc $b({\bf 0}, \pow{j}{i}w_j)$. The upper bound is found by simply ignoring the interference from this BS. These bounds will be used later in this section to derive   tight bounds on coverage probability. 
\begin{cor}
Conditional Laplace transform of $\ncalI_{o(j,0)}$ given $\ncalV_0=v_0$ at $s=\frac{\tau w_j^{\alpha}}{P_j}$ is bounded by 
\begin{align}
\frac{1}{1+\tau}\leq \ncalL_{\ncalI_{o(j,0)|\ncalV_0}}\left(\frac{\tau w_j^{\alpha}}{P_j}\bigg|v_0\right)\leq 1.\label{eq::Interference_cc_lower_bound}
\end{align}
\label{corr::interference_cc_lower_bound}
\end{cor}
\begin{proof}

Following from Eq.~\ref{eq::interferece_cc_intermediate_step}: 
\begin{align*}
\ncalL_{\ncalI_{o(j,0)|\ncalV_0}}(s|v_0)&=\nbbE_{Y_0}\big[\Scale[0.9]{\frac{1}{1+sP_iv_0Y_{0}^{-\alpha}}|Y_0>v_0^{\frac{1}{\alpha}}\pow{j}{i}w_{j}}\big]\\
&\geq \frac{1}{1+sP_iv_0y_{0}^{-\alpha}}\bigg|_{y_0=v_0^{\frac{1}{\alpha}}\pow{j}{i}w_{j}}.
\end{align*}
Substitution of $s=\frac{\tau w_j^{\alpha}}{P_j}$ gives the final result. The upper bound can be obtained by 
\begin{align*}
\ncalL_{\ncalI_{o(j,0)|\ncalV_0}}(s|v_0)=&\nbbE_{Y_0}\big[\Scale[0.9]{\frac{1}{1+sP_iv_0Y_{0}^{-\alpha}}}|Y_0>v_0^{\frac{1}{\alpha}}\pow{j}{i}w_{j}\big]
\\& \leq \lim\limits_{y_0\to\infty}\frac{1}{1+sP_iv_0y_{0}^{-\alpha}}=1.
\end{align*}
\end{proof}
\begin{lemma}
\label{lemma::interferecne_closed}
 Given the typical user of $\Phi_i^{u}$ connects to any  BS $\in \Phi_j$ at a distance $W_j=w_{j}$, $\forall j\in \ncalK_1$,  {the} Laplace transform of $\ncalI_{c(k)}$ at $s=\frac{\tau w_{j}^\alpha}{P_j}$ is
 \begin{align}
&\ncalL_{\ncalI_{c(k)}}\bigg(\frac{\tau w_{j}^{\alpha}}{P_j}\bigg)=\exp\bigg(-\pi\lamc_k \H(\pow{j}{k}w_{j})^2\bigg), 
\label{eq::I_c_laplace}
\end{align}
where
 \begin{align}
\H={\tau}^{2/\alpha}\frac{2\pi\csc(\frac{2\pi}{\alpha})}{\alpha}.\label{eq::H}
\end{align}
 %
\end{lemma}
\begin{proof}
The proof of this fairly well-known result follows in the same lines as that of Lemma~\ref{lemma::L_I_0(j,k)}, with the only difference being the fact that $\ncalI_{c(k)}$ is independent of  $\serv_j$ and hence, the lower limit of the integral in Eq.~\ref{eq::lemma_3_intermediate} will be zero.
The final form can be obtained by some algebraic manipulations and using the properties of Gamma function \cite[Eq.~3.241.2]{zwillinger2014table}.
\end{proof}
The expressions of Laplace transforms  of interference derived in the above three Lemmas are substituted in Eq.~\ref{eq::coverage_theorem_sh} to get the coverage probability. The results for no shadowing is readily obtained by  putting $\ncalV_k\equiv 1$, which omits the final deconditioning step with respect to $\ncalV_0$.
\begin{cor}[No shadowing] Under the assumption of no shadowing, the  coverage probability of a typical user belonging to $\Phi^u_{i}$ can be expressed as 
\begin{equation}
 \pc^{(i)}=\ncalA_{0}\pc_{0}^{(i)}+\sum\limits_{j=1}^K\ncalA_{j}\pc_{j}^{(i)},\quad \text{with}\label{eq::coverage_no_shadowing} 
 \end{equation}
 \begin{align} \label{eq::p_c_0_compact}
& \pc_{0}^{(i)}=
\frac{1}{\ncalA_{0}}\int\limits_{w_0>0}\exp\bigg(-\frac{\tau N_0 w_0^{\alpha}}{P_0}-\pi\sum\limits_{k=1}^K\pow{0}{k}^{2}\times\notag\\&\bigg(\lambda_k(\G+1)+\lambda_k'\H\bigg)w_{0}^2\bigg)\: f_{Y_0}(w_0)\:{\rm d}w_{0},\\
  \pc_{j}^{(i)}&=\frac{2\pi\lambda_j}{\ncalA_{j}}\int\limits_{w_{j}>0}\exp\bigg(-\frac{\tau N_0 w_j^{\alpha}}{P_j}-\pi\sum\limits_{k=1}^K\pow{j}{k}^2\times\notag\\&\bigg(\lambda_k(\G+1)+\lambda_k'\H\bigg)w_{j}^2\bigg)
\notag\\&\times\int\limits_{y_0>\pow{j}{i}w_{j}}{\frac{f_{Y_0}(y_{0})}{1+\tau(\frac{y_{0}}{\pow{j}{i}w_{j}})^{-\alpha}}} {\rm d}y_0 \:w_{j}\:{\rm d}w_{j},
\label{eq::p_c_j_compact}
 \end{align}
 where $\ncalA_j$ is the association probability to $\Phi_j^{(\rm BS,o)}$ given by: $ \ncalA_j=\nbbE_{Y_j}\prod\limits_{\substack{k=0\\k\neq j}}^K\overline{F}_{Y_k}(\pow{j}{k}Y_j), \ \forall k\in \ncalK_1.$
 
%
%
%
%
%
%
%
%
%
\label{corr::coverage_no_shadow_exact}
  \end{cor}
 Note that the PDF and CCDF of $Y_k$ for $k\in\ncalK$ can be obtained by replacing $\lam_k$ ($\lamc_k$) by $\lambda_k$ ($\lambda_k'$) in Eq.~\ref{eq::nn}.
 \subsection{Bounds on Coverage Probability} In this section, we derive upper and lower bounds on  coverage probability $\pc^{(i)}$ by using the results obtained in Corollary~\ref{corr::interference_cc_lower_bound}. 
  \begin{prop}[Bounds on Coverage] 
\label{prop::bound_cov}
The conditional per-tier coverage probability for $j\in \ncalK$ can be bounded as $\pcL\leq\pc_{j|\ncalV_0}^{(i)}\leq \pcU$, where
\begin{align}
\pcU&=\int\limits_{\scriptscriptstyle w_j>0}\Scale[0.95]{\exp\bigg(-\frac{\tau N_0 w_j^{\alpha}}{P_j}\bigg)}\prod\limits_{k =1}^{K}\scalebox{0.95}{$\ncalL_{\ncalI_{o(j,k)}}\left(\frac{\tau w_j^{\alpha}}{P_i}\right)$}\scalebox{0.95}{$\ncalL_{\ncalI_{c(k)}}\left(\frac{\tau w_j^{\alpha}}{P_j}\right)$}\notag\\&\qquad\qquad\times\fservcond{j}(w_j|v_0){\rm d}w_j,\quad\text{and}
\label{eq::per_tier_cov_upper_bound}
\end{align} 
\begin{align}
\pcL&=\frac{1}{1+\tau}\int\limits_{\scriptscriptstyle w_j>0}\Scale[0.95]{\exp\bigg(-\frac{\tau N_0 w_j^{\alpha}}{P_j}\bigg)}\prod\limits_{k =1}^{K}\scalebox{0.95}{$\ncalL_{\ncalI_{o(j,k)}}\left(\frac{\tau w_j^{\alpha}}{P_i}\right)$}\notag\\&\qquad\qquad\times\scalebox{0.95}{$\ncalL_{\ncalI_{c(k)}}\left(\frac{\tau w_j^{\alpha}}{P_j}\right)$}\fservcond{j}(w_j|v_0){\rm d}w_j.
\label{eq::per_tier_cov_lower_bound}
\end{align}
Hence, from Eq.~\ref{eq::coverage_theorem_coverage}, coverage probability $\pc^{(i)}$ can be bounded by
\begin{align*}
 \pc^{(i), { L}}\leq\pc^{(i)}\leq \pc^{(i), { U}},
\end{align*}
where
\begin{align*}
  \pc^{(i), { L}}&=\nbbE_{\ncalV_0}{\ncalA_{0|\ncalV_0}\pc^{(i)}_{j|\ncalV_0}+\sum\limits_{j=1}^K\ncalA_{j|\ncalV_0}\pcL}],\\
  \pc^{(i), { U}}&=\nbbE_{\ncalV_0}{\ncalA_{0|\ncalV_0}}\pc^{(i)}_{j|\ncalV_0}\notag+\sum\limits_{j=1}^K\ncalA_{j|\ncalV_0}\pcU].
\end{align*}
\label{prob::bound}
\end{prop}
\begin{proof}
Using Corollary~\ref{corr::interference_cc_lower_bound}, bounds on $\pc_{j|\ncalV_0}^{(i)}$ can be directly obtained by substituting the  bounds on $\ncalL_{\ncalI_{o(j,0)|\ncalV_0}}(\cdot)$ from Eq.~\ref{eq::Interference_cc_lower_bound} in Eq.~\ref{eq::coverage_theorem_sh}.
\end{proof}
\begin{remark}\label{rem::intuition}
The intuition behind the upper and lower bound on coverage probability is underestimating and overestimating the interference from the BS at cluster center when the typical user does not connect to it (refer to Corollary~\ref{corr::interference_cc_lower_bound}). Given that the user connects to some tier $j\in\ncalK$, no BS including that at the cluster center (equivalently the BS of tier $0$) must lie beyond the exclusion disc of radius $\bar{P}_{ji}w_j$, $w_j$ being the serving distance. Upper bound on coverage will be obtained if the interfering BS at cluster center is pushed away to infinity and lower bound is obtained if it is assumed to be located on the boundary of the exclusion disc.   
\end{remark}
For no shadowing, we can write simpler expressions for the upper and lower bound of $\pc^{(i)}$. This result is presented in the following Proposition. 

\begin{prop}[Bounds on Coverage: No Shadowing]$\pc_j^{(i)}$ can be bounded by
\begin{align}
&\frac{2\pi\lambda_j}{\ncalA_{j}(1+\tau)}\ncalH_j\leq\pc_{j}^{(i)}\leq
\frac{2\pi\lambda_j}{\ncalA_{j}}\ncalH_j,\label{eq::coverage_no_shadow_bounds}
\end{align}\label{prop::bound_no_shadow}
where $\ncalH_j = \int\limits_{{w_j>0}}\exp\bigg(-{\frac{\tau N_0 w_j^{\alpha}}{P_j}}-w_{j}^2\bigg)\overline{F}_{Y_0}(\pow{j}{0}w_j)  \:w_{j}\:{\rm d}w_{j}$. The upper and lower bounds on $\pc^{(i)}$ can be obtained by substituting $\pc^{(i)}_j$ with its upper and lower bounds in Eq.~\ref{eq::coverage_no_shadowing}.
\end{prop}  
It can be readily observed from Proposition~\ref{prop::bound_cov} and \ref{prop::bound_no_shadow} that the bounds on per-tier coverage probability (for $j\in \ncalK$) are simplified expressions due to the elimination of one integration by bounding  $\ncalL_{\ncalI_{o(j,0)}}(\cdot)$. In the following Propositions, we present closed form bounds on coverage probability under no shadowing for Thomas and \matern\  cluster processes in an interference limited network. The tightness of the proposed bounds will be investigated in Section~\ref{sub::tightness_of_bounds}. 
\begin{prop}[Bounds on Coverage: Thomas cluster process]\label{prop::coverage_thomas_no_shadow_bound}
For an interference limited network ($N_0=0$), when $\Phi^u_{i}$ is Thomas cluster process, $\pc^{(i)}$ can be bounded by
\begin{align}
&\frac{\lambda_0}{\ncalM_0}+\frac{1}{1+\tau}\sum\limits_{j=1}^K\frac{\lambda_j}{\ncalM_j}\leq\pc^{(i)}
\leq \sum\limits_{j=0}^K\frac{\lambda_j}{\ncalM_j},
\label{eq::bound_cov_thomas}
\end{align}
where $\ncalM_j=\lambda_0+\sum\limits_{k=1}^K\pow{j}{k}^2(\lambda_k(\G+1)+\lambda_k'\H)$. 
\end{prop}
\begin{proof}
 {
  $\pc_0^{(i)}$ can be obtained by substituting $f_{Y_{0}}(\cdot)$ by Eq.~\ref{eq::dist_thomas} in Eq.~\ref{eq::p_c_0_compact}. This gives the first term $\lambda_0/\ncalM_0$ in the expressions of the two bounds.  $\pc_j^{(i),L}$ and $\pc_j^{(i),U}$ can be obtained by substituting $\overline{F}_{Y_0}(\cdot)$ from Eq.~\ref{eq::dist_thomas_ccdf} to Eq.~\ref{eq::coverage_no_shadow_bounds}. The result follows from evaluation of the integrals. }
\end{proof}
\begin{prop}[Bounds on Coverage: \matern\ cluster process]\label{prop::coverage_matern_no_shadow_bound}
For an interference limited network ($N_0=0$), when $\Phi^u_{i}$ is \matern\  cluster process, $\pc^{(i)}$ can be bounded by
\begin{align}
&\ncalP_0+\frac{1}{1+\tau}\sum\limits_{k=1}^K\ncalP_j\leq\pc^{(i)}
\leq \ncalP_0+\sum\limits_{k=1}^K\ncalP_j,
\label{eq::bound_cov_matern}
\end{align}
where $\ncalP_j=\ncalA_j\pc_j^{(i)}$ which can be obtained by replacing $\ncalZ_j$ by $\pi\sum\limits_{k=1}^K\pow{j}{k}^2(\lambda_k(\G+1)+\lambda_k'\H)$, $\lam_k\ (\lam_k)$ by $\lambda_k\ (\lambda'_k)$ and putting $\ncalV_0\equiv 1$ in the expression of conditional association probability when $\Phi^u_{i}$ is a \matern\ cluster process (Eq.~\ref{eq::association::matern}).     
\end{prop}
\begin{proof}{The proof follows  the similar lines of the previous one, except the substitution of  $f_{Y_0}(\cdot)$ and $\overline{F}_{Y_0}(\cdot)$ by the PDF and CCDF of $Y_0$ for \matern\ cluster process mentioned in Eq.~\ref{eq::dist_matern} and Eq.~\ref{eq::dist_matern_ccdf} respectively. The integrals will be exactly in the similar forms as those appearing in the proof of Corollary~\ref{cor::association_matern} and the final result follows on the same line of the proof.  }
\end{proof}

\subsection{Asymptotic Analysis of Coverage}
\label{subsec::asymptotic}
In this section, we examine the limiting behaviour of the coverage probability expressions with respect to the cluster size. As the cluster size increases, {the typical user is pushed away from the cluster center, which reduces its association probability with the BS located at its cluster center}. Also the interference and coverage provided by the BS at the cluster center will be diminished due to reduced received signal power from this BS.

As the cluster size increases, let us assume that the distance of the typical user from the cluster center $Y_0$ is scaled to $Z=\zeta Y_0$, where $\zeta>1$ is the scaling factor. Then, the PDF  of $Z$ is $f_Z(z)= \frac{1}{\zeta}f_{Y_0}(z/\zeta)$. The scaling of the distance from the cluster center and increasing the cluster size are equivalent, for instance,  when $\Phi_i^u$ is a \matern\ cluster process,  $0<Y_0<\Ri$ implies that $0<Z<\std \Ri$. When $\Phi^u_i$ is Thomas cluster process,  since the users have a Gaussian distribution around cluster center, $99.7\%$ of the total users in cluster will lie within a disc of radius $3 \: \sigma_i$. 
Thus $\sigma_i$ can be treated as the metric of cluster size and $\sigma_i$ scales with $\std$.  In the following lemma, we investigate the limiting nature of coverage as cluster size goes to infinity.   To retain the  simplicity of  expressions, we restrict the following analysis for  interference limited networks ($N_0 = 0$). However, this can be easily extended for $\sinr$-based coverage probability without much effort.    

\begin{lemma}[Convergence]\label{lemma::convergence}
If distance of the a typical user and cluster center $Y_0$ is scaled by $\zeta$ ($\zeta>1$), then the following limit can be established: $\lim\limits_{\substack{\std\to\infty}}\pc^{(i)} = $
\begin{align}
\pc^{\rm(PPP)}\delequal\sum\limits_{j=1}^K\frac{\lam_j}{\sum\limits_{k=1}^K\pow{j}{k}\bigg(\lam_k(\G+1)+\lamc_k\H\bigg)}.
\label{eq::lemma_limit}
\end{align}
\label{lemma::lemma_limit}
\end{lemma}
\begin{proof}
See Appendix~\ref{app::lemma_limit}.
\end{proof}
\begin{remark}\label{rem::assymptote}
{In Lemma~\ref{lemma::convergence}, we formally claim that, irrespective of the distribution of $Y_0$, if the size of the cluster is expanded, the total coverage probability $\pc^{(i)}$ converges to $\pc^{(\rm PPP)}$, i.e., the coverage probability obtained for a typical user under the assumption of PPP distribution of users independent of the BS point processes, which is derived in \cite{xia}.}
\end{remark}

\subsection{Overall Coverage Probability}
The results so far are concerned with the users  belonging to $\Phi_i^{u}$. Recall that in our system model we considered that the users form a mixed point process consisting of $\Phi_i^{u}$ ($i\in\ncalB$)  and $\Phi^{u\rm(PPP)}$. So the overall coverage probability will be a combination of all these individual coverage probabilities $\pc^{(i)}$ ($i\in\ncalB$) and also $\pc^{\rm(PPP)}$ corresponding to the users in $\Phi^{u\rm(PPP)}$,  which are distributed independently of the BS locations. 
The overall user point process can be expressed as $\Phi^u\equiv\Phi^{u\rm(PPP)}\cup\left(\bigcup\limits_{i\in\ncalB}\Phi^u_{i}\right)$. The average number of points of $\Phi^u$ in any given set $A \subset \mathbb{R}^2$ is given by
\begin{align*}
\nbbE(\Phi^u(A))=\nbbE(\Phi^{u\rm(PPP)}(A))+\sum_{i\in\ncalB}\nbbE(\Phi^u_i(A)) ,
\end{align*}
  where  $\nbbE(\Phi^{u\rm(PPP)}(A))=\lambda^{\rm(PPP)}A, \nbbE(\Phi^u_i(A))=N_i\lambda_i A$. To avoid notational complication, we use the symbol $\Phi$ to denote a point process as well as the associated counting measure.   Since each point has an equal chance to be selected as location of the typical user, the probability that a randomly chosen user from $\Phi^u$ belongs to $\Phi^{u\rm(PPP)}$ $(\Phi_i^{u})$, denoted by $p_0$ $(p_i)$ respectively, is
  \begin{align*}
  p_0=\frac{\lambda^{\rm(PPP)}}{\lambda^{\rm(PPP)}+\sum\limits_{j\in\ncalB}N_j\lambda_j} \ \text{and}\ 
p_i=\frac{N_i\lambda_i}{\lambda^{\rm(PPP)}+\sum\limits_{j\in\ncalB}N_j\lambda_j},
\end{align*}   
where  $N_i$ is the average number of users per cluster of $\Phi_i^{u}$ ($i\in\ncalB$).   Now, using these probabilities $p_0$ and $p_i$, the overall coverage probability is formally stated in the next Theorem. 
\begin{theorem}[Overall Coverage Probability]\label{th::mixed} Overall  coverage probability with respect to any randomly chosen user in a $K$-tier HetNet with mixed user distribution is:
\begin{align}
\pc&=p_0\pc^{\rm(PPP)}+\sum\limits_{i\in\ncalB}p_i\pc^{(i)},
\end{align} 
where $\pc^{\rm(PPP)}$ and $\pc^{(i)}$ are given by Eq.~\ref{eq::lemma_limit} and Eq.~\ref{eq::coverage_theorem_coverage}, respectively.
\end{theorem} 
\section{Numerical Results and Discussions}
\label{section::numerical_results}
\subsection{Validation of Results}\label{subsec::validation}
In this section, the analytical results derived so far are validated and key insights for the new HetNet system model with users clustered around BSs are provided.  For the sake of concreteness, we restrict our simulation to two tiers: one macrocell tier ($\Phi_1^{\rm(BS,o)}$) with density $\lambda_1$ with all open access BSs, and one small cell tier ($\Phi_2^{\rm(BS)}$) with a mix of open and closed access BSs. For $\Phi_2^{\rm(BS)}$, the open and closed access BS densities are $\lambda_2$ and $\lambda_2'$, respectively. We choose $\lambda_2 = \lambda_2' = 100\lambda_1 = 100$ BSs per $\pi(500)^2$ m$^2$. We assume the transmit powers are related by $P_1=10^3{P_2}$. The user process is considered to be $\Phi^u_{2}$ only, i.e., a PCP around $\Phi_2^{\rm(BS)}$. For every realization, a BS in the $i^{th}$ tier is randomly selected and location of a typical user is generated according to the density function of {(i)} Thomas cluster process (Eq.~\ref{eq::thomas}), and {(ii)} \matern\  cluster process (Eq.~\ref{eq::matern}).  For shadowing,
 we have chosen log-normal distribution parameters as $\mu_k=0$, $\eta_k=8\ {\rm dB}$, $4\ {\rm dB}$ and $0\ {\rm dB}$ (no shadowing)  for all $k=0,1,2$. 
\begin{figure}[t]
  \centering
  \subfloat[Users in Thomas cluster process\label{fig::comparison_thomas_sh}]{\includegraphics[width=0.4\textwidth]{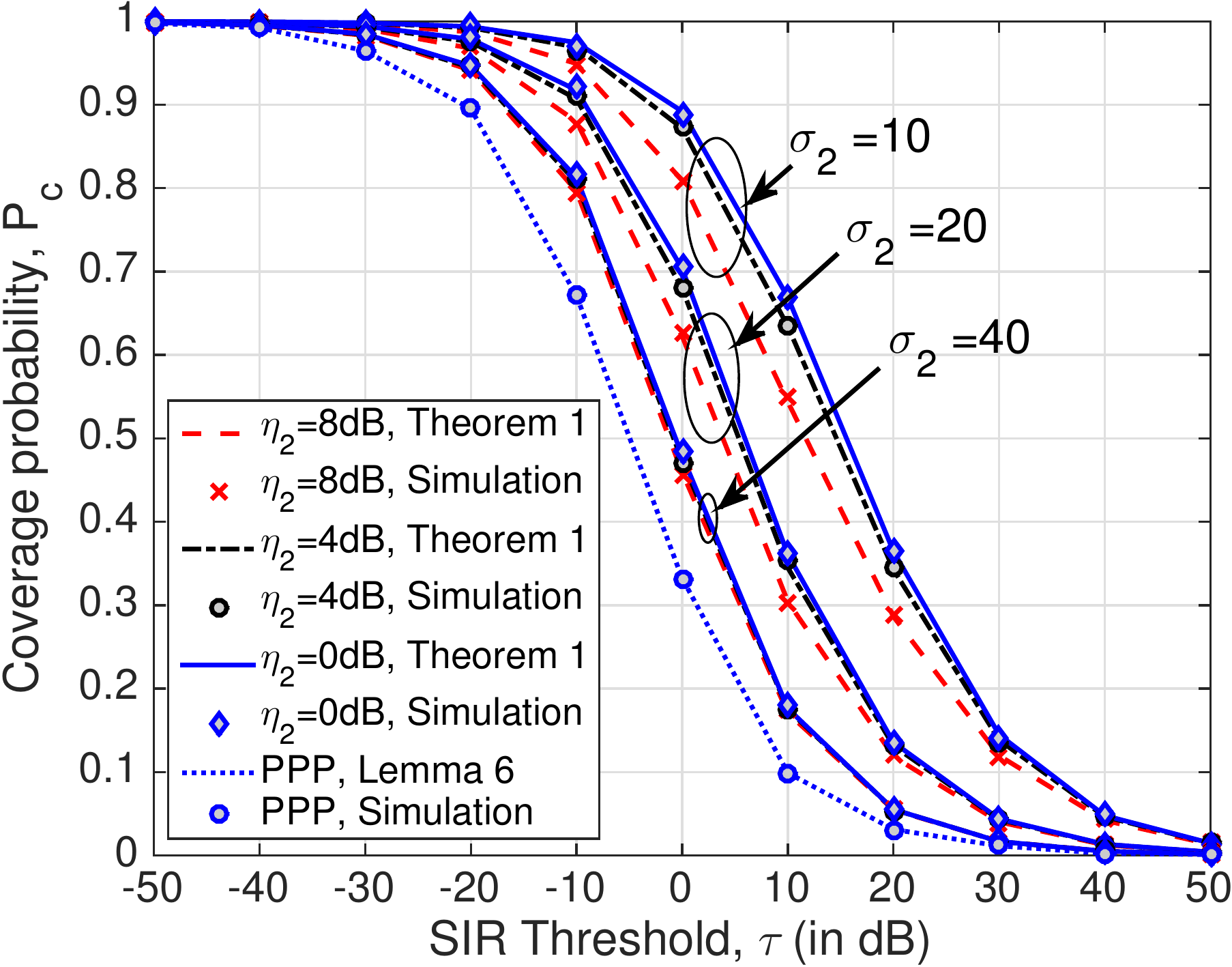}}\quad
  \subfloat[Users in \matern\  cluster process\label{fig::comparison_matern_sh}]{\includegraphics[width=0.4\textwidth]{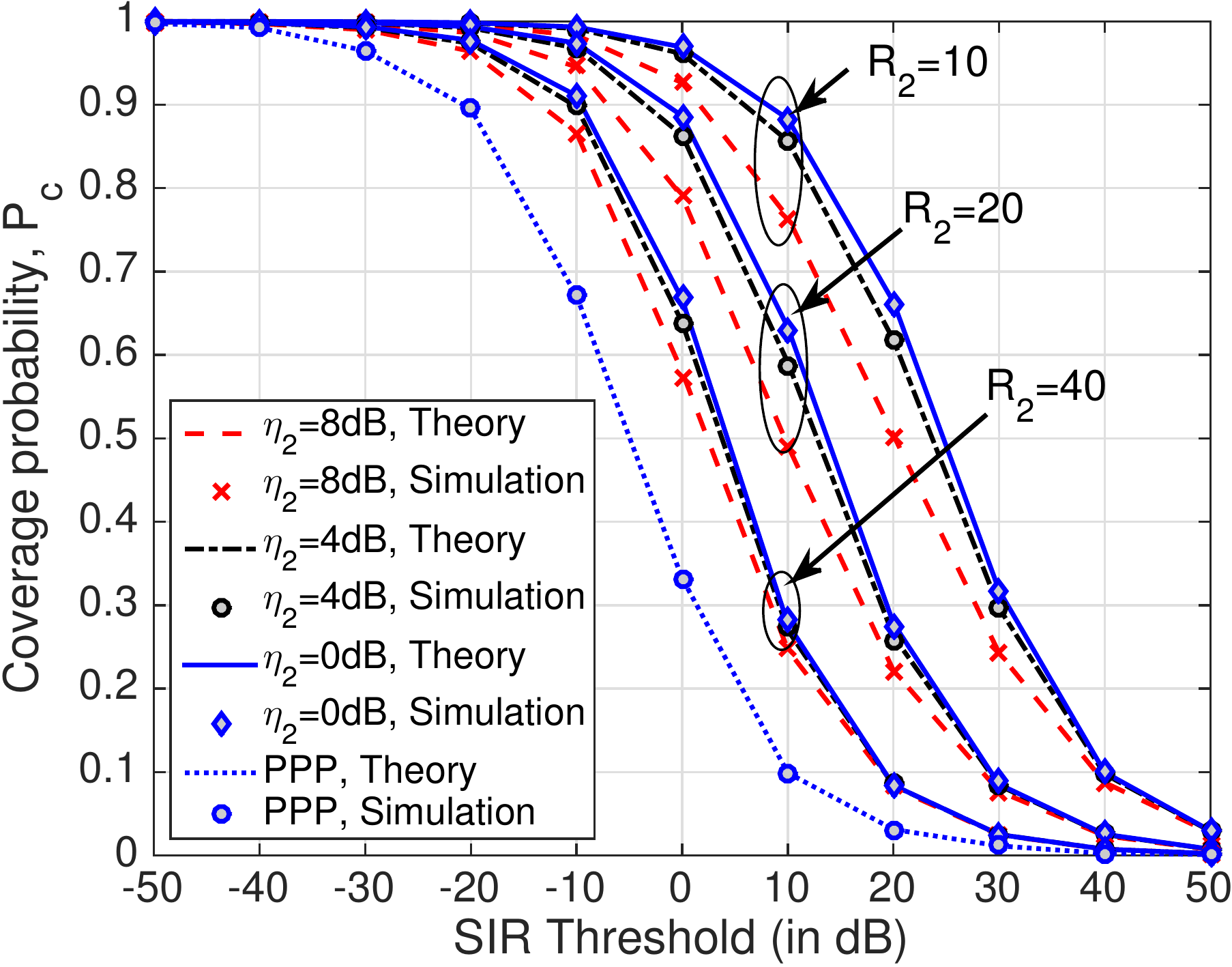}}
  \caption{\small Comparison of coverage probabilities with cluster size for various shadowing environments. The baseline case when the user distribution is a PPP is also included. The lines and markers correspond to the analytical and simulation results, respectively. }
  \label{fig::comparison_coverage} 
\end{figure}
In \figref{fig::comparison_coverage}, the coverage probability ($\pc$, equivalently $\pc^{(2)}$) is plotted  for different values of $\sir$ threshold $\tau$ and cluster size (i.e. different $\sigma_2$-s for Thomas and ${\ncalR}_{2}$-s for \matern\  cluster processes) {for an interference limited network ($N_0 = 0$). The validity of this assumption will be justified in the next subsection.} It can be observed that the analytically obtained results exactly match the simulation results. For comparison, $\pc^{\rm(PPP)}$, i.e., the coverage probability assuming homogeneity of users (i.e., independent PPP assumption) is also plotted. The plots clearly indicate that under clustering, $\pc$ is significantly higher than $\pc^{\rm(PPP)}$ and increases for denser clusters.  Also the convergence towards $\pc^{\rm(PPP)}$ is evident as cluster size increases.  
In \figref{fig::association},  the association probabilities are plotted for different cluster size with $\eta_k=4\ {\rm dB}$. The figure clearly illustrates  that a user is more likely to be served by its cluster center if the distribution is more ``dense'' around the cluster center. As the cluster expands, association probability to the BS at cluster center (equivalently the $0^{th}$ tier) decreases whereas the association probabilities to the other open access tiers increase.   
\begin{figure}[t]
\centering
\subfloat[Users in Thomas cluster process\label{fig::association_thomas}]{\includegraphics[trim={0cm 0cm 0cm 0cm},clip,width=0.4\textwidth]{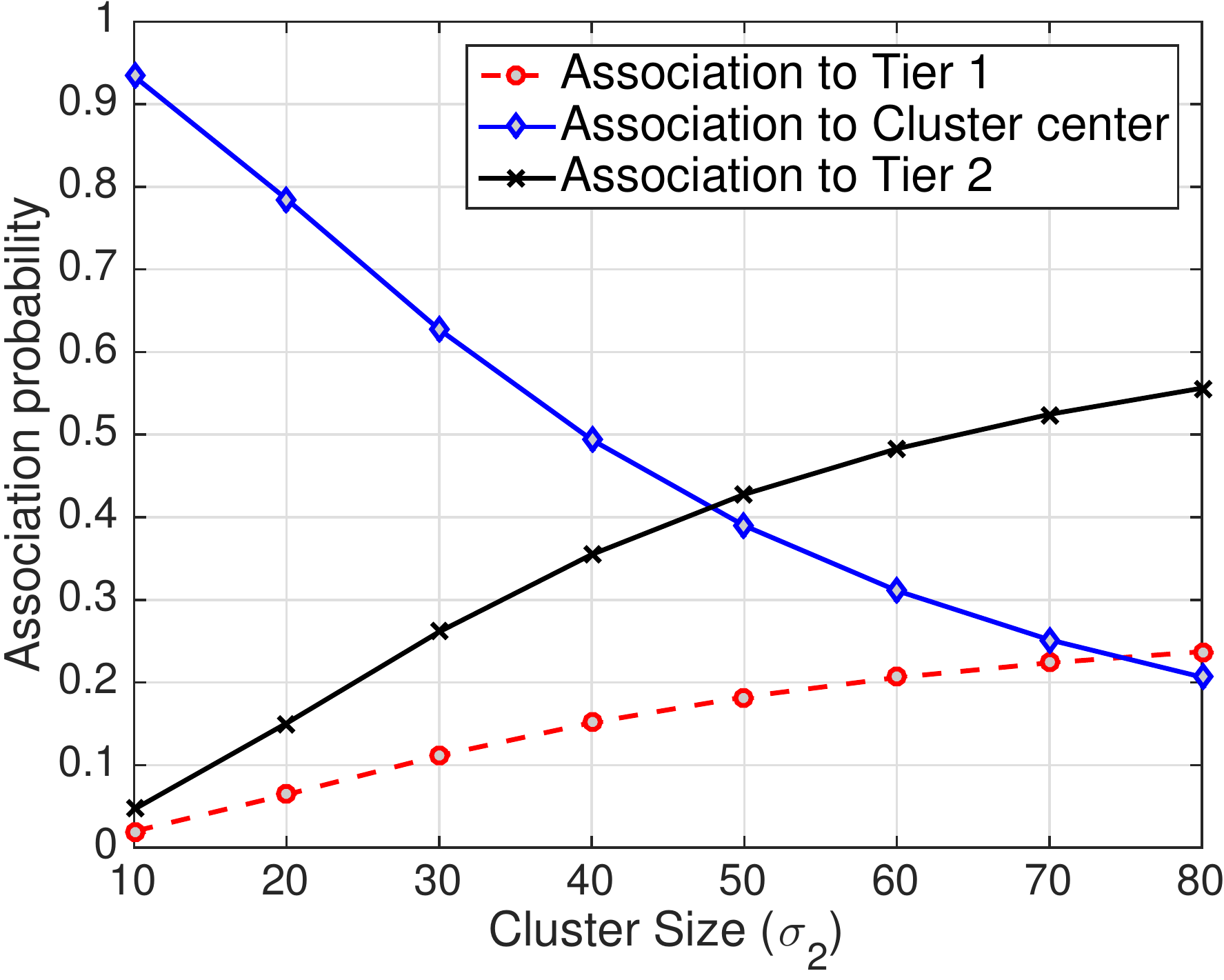}}\quad
\subfloat[Users in \matern\  cluster process\label{fig::association_matern}]{\includegraphics[trim={0cm 0cm 0cm 0cm},clip,width=0.4\textwidth]{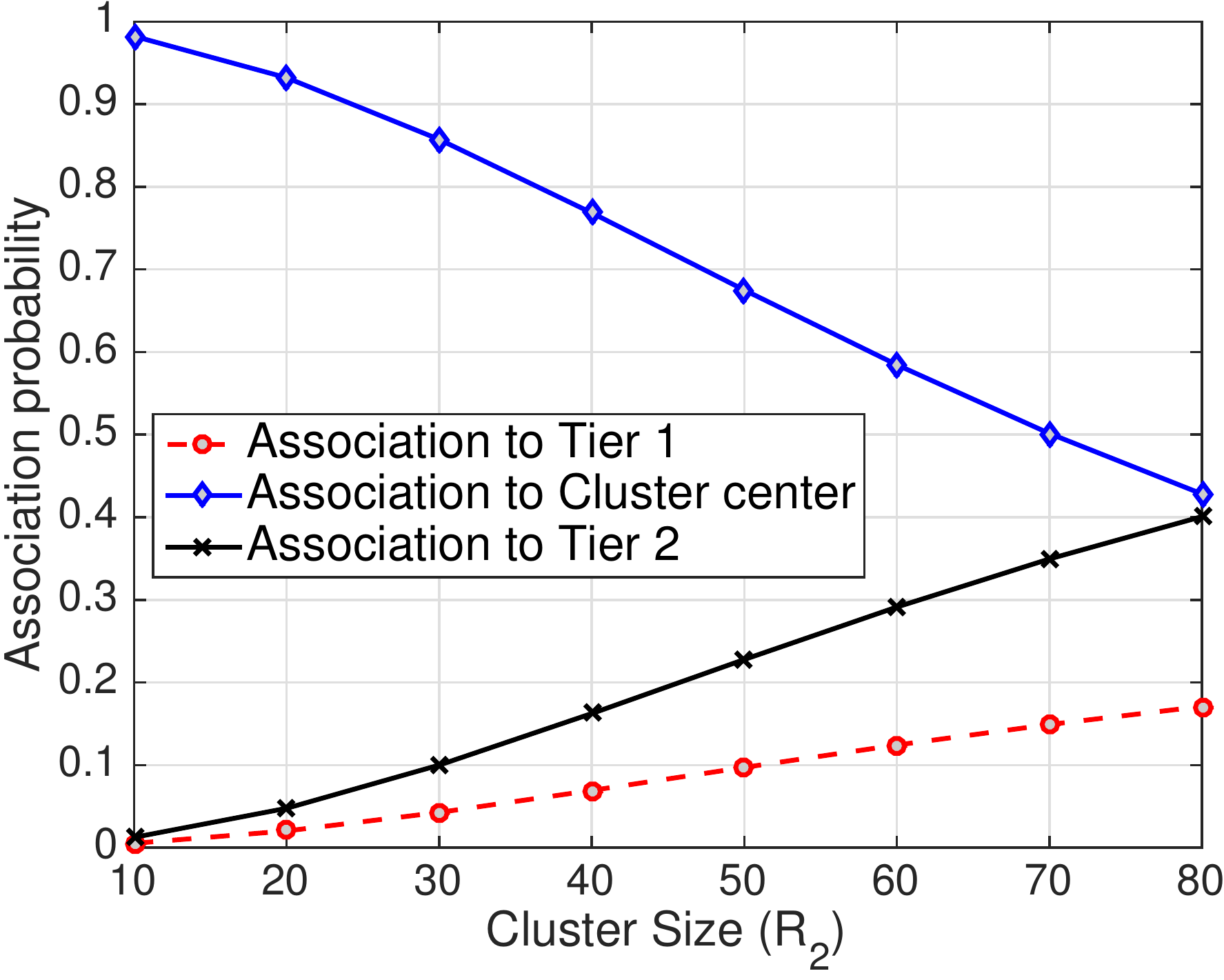}}
\caption{\small Comparison of the association probabilities to the two tiers and the cluster center.}
\label{fig::association}
\end{figure}
\begin{figure}
  \centering
  \subfloat[Users in Thomas cluster process\label{fig::bound_thomas}]{\includegraphics[width=0.4\textwidth]{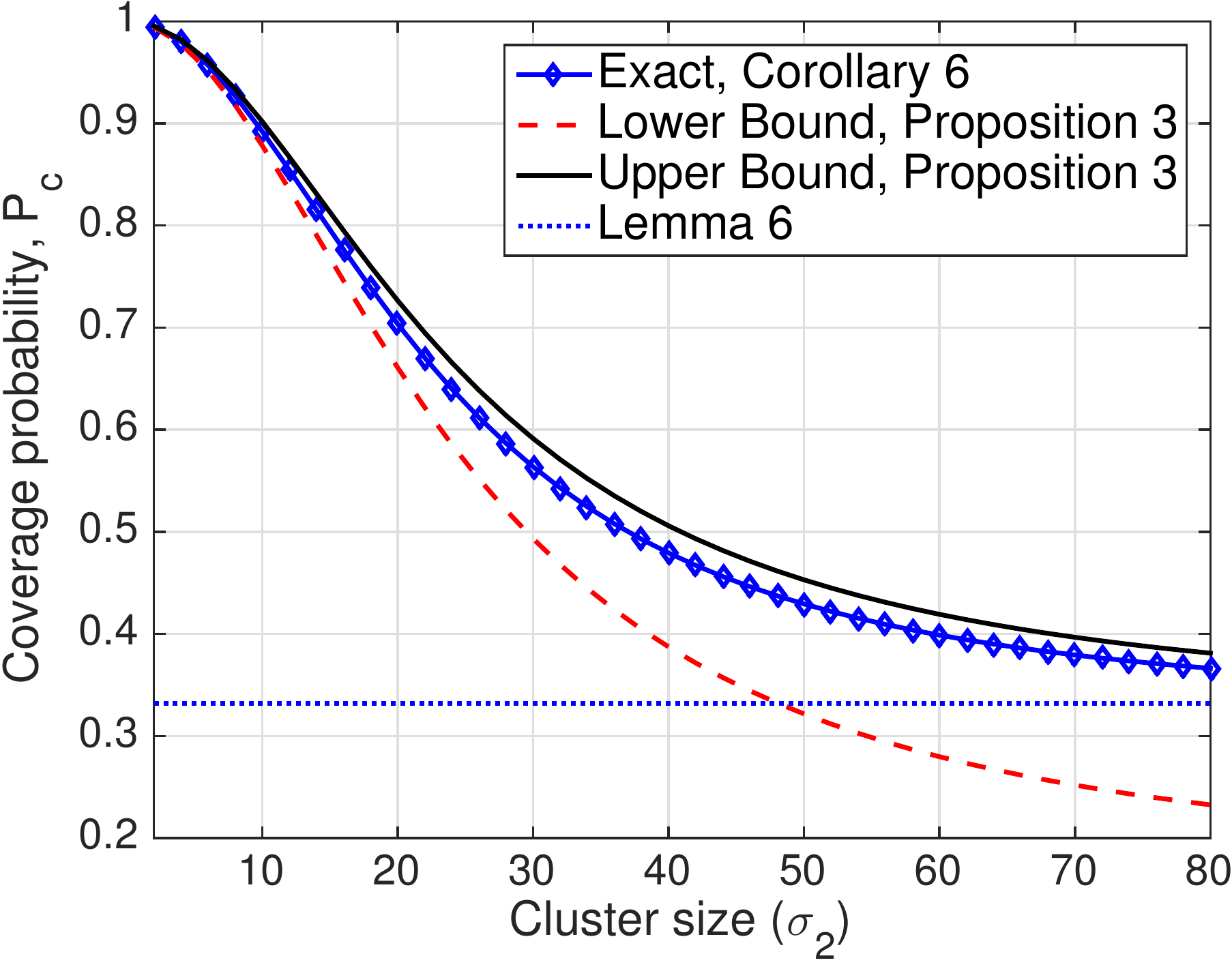}}\quad
  \subfloat[Users in \matern\ cluster process\label{fig::bound_matern}]{\includegraphics[width=0.4\textwidth]{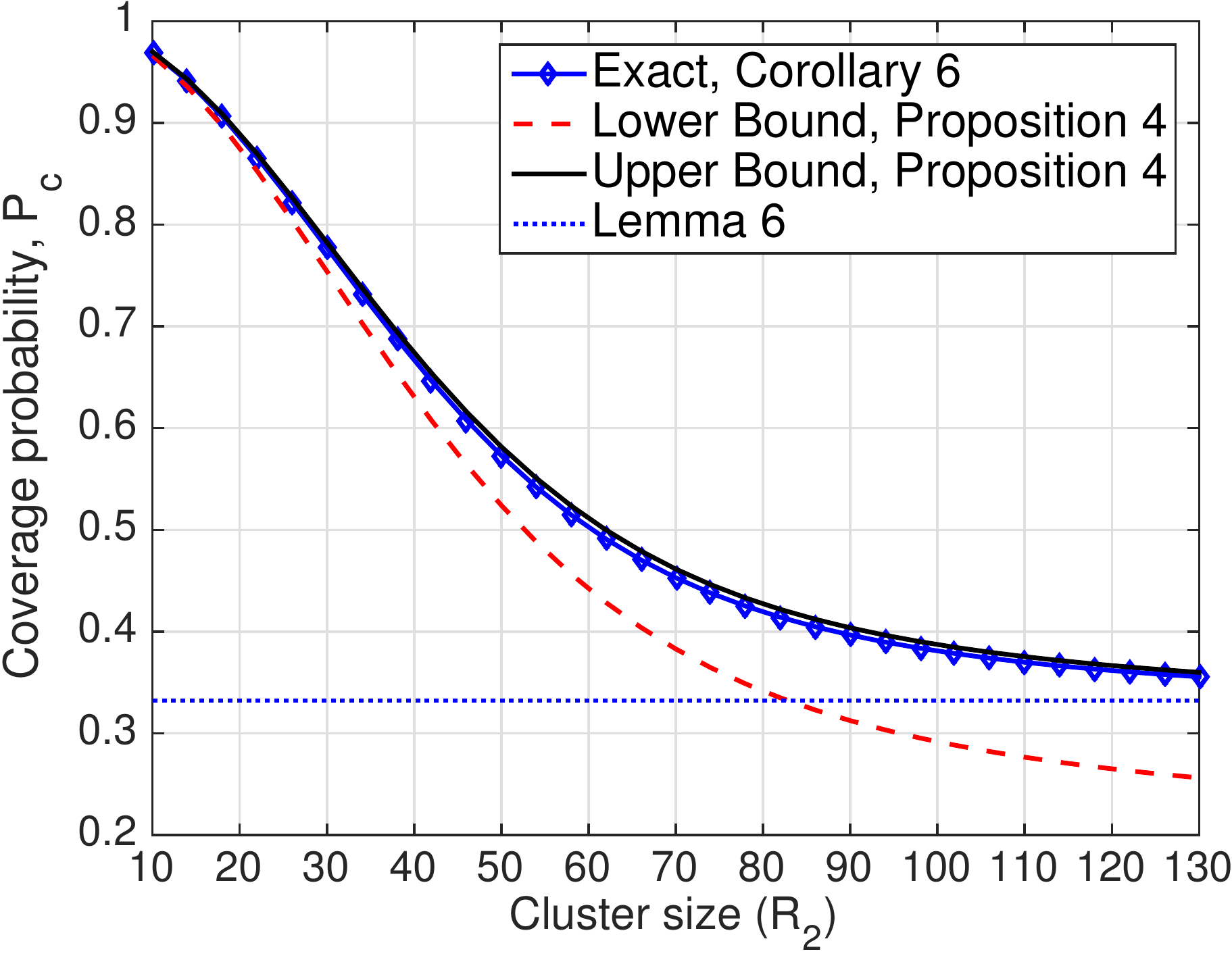}}
  \caption{{\small Inspection of the proposed closed form bound with variation of cluster size for constant $\sir$ threshold, $\tau = 0\  dB$}}
  \label{fig::bound}
\end{figure}
\subsection{Effect of Thermal Noise}\label{subsec::noise}
In this subsection, we investigate the effect of thermal noise on the coverage probability in the two-tier setup described in the previous section. In order to do this, we need to first fix a realistic reference point relative to which the noise variance $N_0$ will be decided. For that we choose the reference signal-to-noise ratio observed at the cell edge of a macrocell. Fixing this value to say $0$ dB we can then calculate the noise variance $N_0$ using the same procedure that we used in \cite[Section V-A]{dhillonHetNet}. Plugging this value in the theoretical results, we compare the coverage probability obtained under this setup with its no-noise counterpart under no shadowing in Fig.~\ref{fig::noise}. As expected, it is observed that the noise does not have any noticeable effect on the coverage probability due to which we will simply ignore the effect of noise in the rest of this section. 
\begin{figure}[t]
\centering
\subfloat[Users in Thomas cluster process\label{fig::noise_thomas}]{\includegraphics[trim={0cm 0cm 0cm 0cm},clip,width=0.4\textwidth]{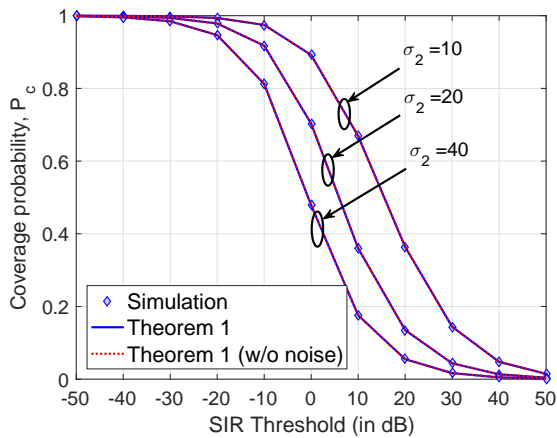}}\quad
\subfloat[Users in \matern\  cluster process\label{fig::noise_matern}]{\includegraphics[trim={0cm 0cm 0cm 0cm},clip,width=0.4\textwidth]{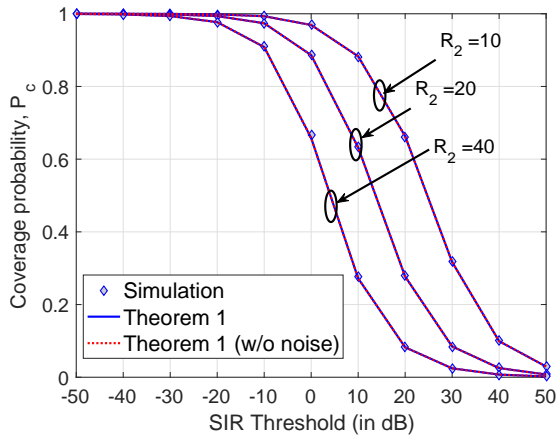}}
\caption{\small Comparison of coverage probabilities with and without thermal noise under no shadowing.}
\label{fig::noise}
\end{figure}
\subsection{Tightness of the Bounds}\label{sub::tightness_of_bounds}
In Proposition~\ref{prop::bound_cov}, we derived upper and lower bounds on $\pc^{(i)}$. We found that for no shadowing, these bounds reduce to closed form expression when $\Phi^u_{i}$ is Thomas or \matern\  cluster process (Propositions~\ref{prop::coverage_thomas_no_shadow_bound} and  \ref{prop::coverage_matern_no_shadow_bound}). In \figref{fig::bound}, we plot these upper and lower bounds on $\pc$. Recall that the lower bound was obtained by placing the BS of the cluster-center (in the representative cluster) on the boundary of the exclusion disc when the typical user connects to  other BSs and the upper bound was found by simply ignoring the interference from this BS (see Corollary~\ref{corr::interference_cc_lower_bound} for details). We observe that the lower bound becomes loose as the cluster size increases and for large user clusters, $\pc^{\rm(PPP)}$ becomes  tighter lower bound. 
This is because the interference from the cluster center is significantly overestimated by placing the BS of the cluster center at the boundary of the exclusion zone. The upper bound remains tight for the entire range of cluster sizes. This can be explained by looking at the cases of small and large clusters separately. For small clusters, the typical user will likely connect to the BS at its cluster center most of the time and hence the interference term in question (Laplace transform of interference from the cluster center; see Corollary~\ref{corr::interference_cc_lower_bound}) will not even appear in the coverage probability expression. On the other hand, for large clusters, the interference from the BS at the cluster center of the representative cluster will be negligible compared to the other interference terms due to large distance between the typical user and this BS. 
\subsection{Power Control of small cell BSs}\label{subsec::power_control}
{If $\Phi^u$ is a PPP independent to BS locations, then $\pc^{\rm(PPP)}$ is independent of the BS transmission power and it predicts that no further gain in coverage can be achieved by increasing $P_2/P_1$ (for interference-limited HetNet consisting of open access BSs under the assumption that the target SIR is the same for all the tiers)~\cite{dhillonHetNet}. In the typical two tier HetNet setup described in Section~\ref{subsec::validation}, we set the density of closed access tiers, $\lambda_1'=\lambda_2'=0$ and  vary  $P_1$ keeping  $P_2$ constant and plot $\pc$ in Figs.~\ref{fig::power_control_thomas} and \ref{fig::power_control_matern} and fix $\tau = 0\ {\rm dB}$.}  It is evident that $\pc$ improves significantly with $P_2/P_1$. In the figures, we can identify three regions of $\pc$: (i) For lower value of $P_2/P_1$, $P_c$ is close to $\pc^{\rm(PPP)}$, (ii) $\pc$ is enhanced as $P_2/P_1$ increases since the user is likely to be served by the cluster center, (iii) if $P_2/P_1$ is further increased, $\pc$ is saturated since association probability to other BSs will diminish. Again, the gain of $\pc$ is stronger for denser clusters. Thus, coverage gain can be harnessed by  increasing the transmit powers of small cell BSs in a certain range. 
\begin{figure}
\centering
\subfloat[Users in Thomas cluster process\label{fig::power_control_thomas}]{\includegraphics[trim={0cm 0cm 0cm 0cm},clip,width=0.42\textwidth]{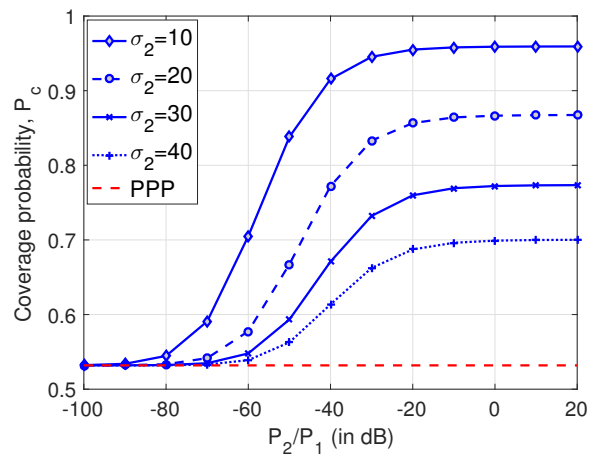}}
\hspace{.5cm}
\subfloat[Users in \matern\  cluster process\label{fig::power_control_matern}]{\includegraphics[trim={0cm 0cm 0cm 0cm},clip,width=0.42\textwidth]{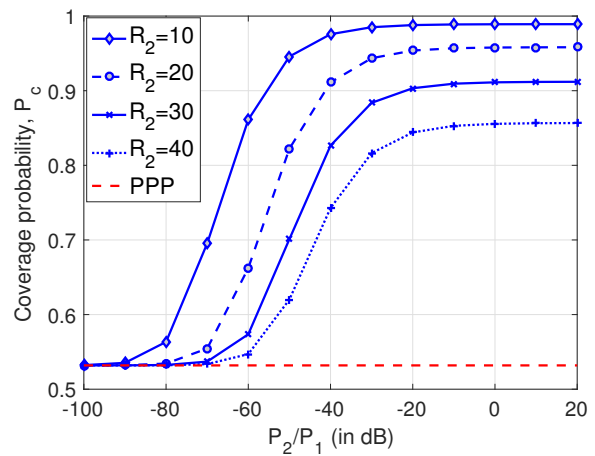}}
\caption{\small Effect of increasing small cell power on coverage.}
\label{fig::power_control}
\end{figure}
\section{Conclusion}
While random spatial models have been used successfully 
to study various aspects of HetNets in the past few years,
quite remarkably all these works assume the BS and user
distributions to be independent. In particular, the analysis is
usually performed for a typical user whose location is sampled
independently of the BS locations. This is clearly not the
case in current capacity-driven user-centric deployments where the BSs are
deployed in the areas of high user density. This paper presented
a  comprehensive analysis of such user-centric HetNet deployments in which the user and BS locations are naturally correlated. In particular, modeling the user
locations as a general Poisson cluster process, with BSs being
the cluster centers, we have developed new tools leading to
tractable results for the downlink coverage probability of a
typical user. We have specialized the results for the case of Thomas cluster process in which
the users are Gaussian distributed around BSs, and \matern\  cluster process where the users are
uniformly distributed inside a disc centered at the BS. We have also examined the bounds and
the limiting nature of the coverage probability as cluster size goes to infinity. We have  derived
the overall coverage probability for a mixed user distribution containing users uniformly distributed and clustered around small cell BSs. 
Overall, this work opens up a new dimension in the HetNet analysis by
providing tools for the analysis of non-uniform user distributions correlated to the BS locations.

This work has numerous extensions. From the system model side, one can perform measurement campaigns to characterize the nature of different user clusters at hotspot locations, such as restaurants, sports bars, and airports. Various cluster process models can then be fitted to this real-world data to obtain accurate user location models, which can then be used for more accurate performance analyses.  From the analytical point of view, an immediate extension is to perform the rate analysis and study the effect of traffic offloading from macrocells to small cells in the current setup. Also, in this work, we assumed the BS locations to be independent from each other. This may not always be the case. For instance, small cells, such as picocells, may not be deployed close to macrocells. Such dependencies have been modeled recently in~\cite{SpatiotemporalHaenggi2015,AfZhMuCh2015} by modeling the BS distribution asusing Poisson Hole Process~\cite{yazdanshenasan2016poisson}. Also, the smallcells may be densely deployed in user hotspot zones and the spatial distribution can be modeled by PCP ~\cite{AfshDhiHetNet2016}. Other considerations, such as device-to-device (D2D) communication in clusters can also be incorporated in this model~\cite{AfshDhiAsilomar2015}.

\appendix
\subsection{Proof of Lemma~\ref{lemma::association_prob_sh}}
\label{app::association_prob_sh}
  According to the definition of $S_{\Phi_j}$ in Eq.~\ref{eq::def_s_phi}, we can write from Eq.~\ref{eq::association_prob_defn}, 
\begin{align}
&\ncalA_{j|v_0}=\nbbE_{\mathbf{R}}\left[\bigcap\limits_{{k\in\ncalK_1\setminus \{j\}}}\mathbf{1}\left(R_k>\pow{j}{k}R_j\right)|\ncalV_0=v_0\right]\notag\\ 
&\myeq{a}\nbbE_{R_j}\prod_{\substack{k=0\\k\neq j}}^K\nbbP\left(R_k>\pow{j}{k}R_j|v_0\right)=\nbbE_{R_j}\prod_{\substack{k=0\\k\neq j}}^K\overline{F}_{R_k}(\pow{j}{k}R_j|v_o), \label{eq::proof_lemma_1}
\end{align}
where (a) comes from the fact that $\Phi_k$-s are independent, hence are $R_k$-s. 
For the rest of the proof, we need to consider the two cases of $j=0$ and $j\neq 0$ separately. Note that only the RV $R_0$ among all $R_j$-s is the function of $\ncalV_0$. For \textbf{case 1:} $j=0$ 
\begin{align*}
\ncalA_{j|v_0}=\nbbE_{R_0}\prod_{k=1}^K\overline{F}_{R_k}(\Scale[0.9]{\pow{j}{k}R_0|v_0})=\nbbE_{Y_0}\prod_{k=1}^K
\overline{F}_{R_k}(\pow{0}{k}v_0^{-\frac{1}{\alpha}}Y_0),  
\end{align*}
{and for \textbf{case 2:} $j\in \ncalK$,
\begin{align*}
 \ncalA_{j|v_0}&=\nbbE_{\mod_j}\bigg[\prod\limits_{\substack{{k=0}\\{k\neq j}}}^K\nbbP(\mod_k>\pow{j}{k}\mod_j|v_0)\bigg]\\&=\nbbE_{\mod_j}\bigg[\nbbP(v_0^{-\frac{1}{\alpha}}Y_0>\pow{j}{0}\mod_j)\prod\limits_{\substack{{k=1}\\{k\neq j}}}^K\nbbP(\mod_k>\pow{j}{k}\mod_j)\bigg]\\&=\nbbE_{\mod_j}\bigg[\overline{F}_{Y_0}(v_0^{\frac{1}{\alpha}}\pow{j}{0}R_j)\prod\limits_{\substack{k=1\\k\neq j}}^K\overline{F}_{\mod_k}(\pow{j}{k}R_j)\bigg].
 \end{align*}
\subsection{Proof of Corollary~\ref{cor::association_thomas}}
\label{app::association_thomas}
When $j=0$, from Eq.~\ref{eq::association_prob_main_lemma}, we get,
$$\ncalA_{0|v_0}=\int_{y_0>0}\prod\limits_{k=1}^K\overline{F}_{R_k}(\pow{0}{k}v_0^{-\frac{1}{\alpha}}y_0)\:f_{Y_0}(y_0)\:{\rm d}y_0.$$

Substituting $f_{Y_0}(y_0)$ from Eq.~\ref{eq::dist_thomas} and $\overline{F}_{R_k}(\pow{0}{k}v_0^{-\frac{1}{\alpha}}y_0)$ from Eq.~\ref{eq::nn_distribution_ccdf}, we get,
\begin{align*}
&\ncalA_{0|v_0}=\int\limits_{0}^{\infty}\exp\big(-\Scale[0.9]{\pi\sum\limits_{k=1}^{K}\lam_k\pow{0}{k}v_0^{-\frac{2}{\alpha}}y_0^2}\big)
\frac{y_{0}}{\sigma_i^2}\exp\bigg(-\frac{y_0^2}{2\sigma_i^2}\bigg){\rm d}y_0\\&=\frac{\frac{v_0^{\frac{2}{\alpha}}}{2\pi\sigma_i^2}}{\frac{v_0^{\frac{2}{\alpha}}}{2\pi\sigma_i^2}+\sum\limits_{k=1}^K\pow{0}{k}^2\lam_k}.
\end{align*}
Putting $\lam_0=\frac{v_0^{\frac{2}{\alpha}}}{2\pi\sigma_i^2}$, we get the desired result. Note that $\pow{0}{0}=1$.
For $j\neq 0$, 
\begin{align*}
&\ncalA_{j|v_0}=\int\limits_{0}^{\infty}\overline{F}_{Y_0}(v_0^{\frac{1}{\alpha}}\pow{j}{0}R_j)\prod\limits_{\substack{k=1\\k\neq j}}^K\overline{F}_{R_k}(\pow{j}{k}r_j)f_{R_j}(r_j)\:{\rm d}r_j\\&=\int\limits_{0}^{\infty}\exp\left(-\frac{(\pow{j}{0}v_0^{\frac{1}{\alpha}}r_j)^2}{2\sigma_i^2}\right)\exp\bigg(-\pi\sum\limits_{\substack{{k=1}\\{k\neq j}}}^{K}\lam_k\pow{j}{k}^2r_j^2\bigg)
2\pi\lam_j\times\\
&\exp(-2\pi\lam_jr_j^2)\:r_{j}\:{\rm d}r_j=\frac{\lam_j}{\frac{\pow{j}{0}^2}{2\pi\sigma_i^2}+\sum\limits_{k=1}^K\pow{j}{k}^2\lam_k}=\frac{\lam_j}{\sum\limits_{k=0}^K\pow{j}{k}^2\lam_k}.
 \end{align*}
 The last step was derived by putting $\lam_0=\frac{v_0^{\frac{2}{\alpha}}}{2\pi\sigma_i^2}$.
\subsection{Proof of Corollary~\ref{cor::association_matern}}
\label{app::association_matern}
Similar to Corollary~\ref{cor::association_thomas}, for $j=0$, plugging Eq.~\ref{eq::nn_distribution_ccdf} and Eq.~\ref{eq::dist_matern} in Eq.~\ref{eq::association_prob_main_lemma}, we get, 
\begin{align*}
\ncalA_{0|v_0}&= \int\limits_{0}^{\Ri}\exp\bigg(-\pi\sum\limits_{k=1}^{K}\lam_k\pow{0}{k}^2v_0^{-\frac{2}{\alpha}}y_0^2\bigg)\:\frac{2y_0}{\Ri^2}\: {\rm d}y_0\\&
= \frac{v_0^{\frac{2}{\alpha}}}{\Ri^2\ncalZ_0}\bigg(1-\exp\bigg(-{v_0^{-\frac{2}{\alpha}}}\ncalZ_0\Ri^2\bigg)\bigg),
\end{align*}
where $\ncalZ_0=\pi\sum\limits_{k=1}^{K}\lam_k\pow{0}{k}^2$.
Now for $j\in \ncalK$, using Eq.~\ref{eq::nn_distribution}, Eq.~\ref{eq::nn_distribution_ccdf} and Eq.~\ref{eq::dist_matern_ccdf} in Eq.~\ref{eq::association_prob_defn} and proceeding according to the proof of Corollary~\ref{cor::association_thomas}, we get, 
\begin{align*}
&\ncalA_{j|v_0}= 2\pi\lam_j\int\limits_{0}^{\frac{\Ri}{\pow{j}{0}{v_0^{\frac{1}{\alpha}}}
}}\exp\left(-\pi\sum\limits_{k=1}^{K}\lam_k\pow{j}{k}^2r_j^{2}\right)\Scale[0.95]{\frac{\Ri^2-(\pow{j}{0}{v_0^{\frac{1}{\alpha}}}r_j)^2}{\Ri^2}}r_j {\rm d}r_j\\
&=\frac{\pi\lam_j}{{\cal Z}_j}-\frac{\lam_j\pi\pow{j}{0}^2v_0^{\frac{2}{\alpha}}}{\Ri^2{\cal Z}_j^2}\bigg(1-\exp\big(-\frac{{\cal Z}_j\Ri^2}{\pow{j}{0}^2v_0^{\frac{2}{\alpha}}}\big)\bigg),\\&\quad \text{where $\ncalZ_j=\pi\sum\limits_{k=1}^K\lam_k\pow{j}{k}^2$ for $j\in\ncalK$.
}
\end{align*}
\subsection{Proof of Lemma~\ref{lemma::serving_distance_dist_sh}}
\label{app::serving_distance_dist_sh}
The conditional  CCDF of $W_j$ in this case is,  
\begin{align}
 &\nbbP[W_j>w_j|\ncalV_0]=\nbbP[R_j>w_j|S_{\Phi_j},\ncalV_0]=\Scale[0.97]{\frac{\nbbP\left(R_j>w_j|\ncalV_0,S_{\Phi_j}|\ncalV_0\right)}{\nbbP(S_{\Phi_j}|\ncalV_0)}}\notag\\&\myeq{a}\frac{1}{\condA{j}}\prod_{\substack{k=1\\k\neq j}}^K\left[\nbbP(P_j  R_j^{-\alpha}>P_k R_k^{-\alpha}|R_j>w_j, \ncalV_0)\right]\notag 
,\end{align}
where (a) follows from Eq.~\ref{eq::proof_lemma_1}. For 
\textbf{case 1:} when $j=0$, given $\ncalV_0=v_0$, $\nbbP[W_0>w_0|\ncalV_0=v_0]=$
\begin{align*}
 \label{eq::intermediate_cdf}
 &\frac{1}{\ncalA_{0|v_0}}\prod_{k=1}^{K}\nbbP(P_0 v_0 Y_0^{-\alpha}>P_k \mod_k^{-\alpha}|v_0^{-\frac{1}{\alpha}}Y_0>w_0)\\&=\frac{1}{\ncalA_{0|v_0}}\int\limits_{v_0^{\frac{1}{\alpha}}w_0}^{\infty}\prod_{k=1}^{K}\overline{F}_{\mod_k}(\pow{0}{k}v_0^{-\frac{1}{\alpha}}y_0)f_{Y_0}(y_0)\:{\rm d}y_0.
\end{align*}
 Thus, the conditional distribution of $W_0$ is obtained by
  \begin{align*}
 & \fservcond{0}(w_0|v_0)=\frac{\rm d}{{\rm d}w_0}\left(1- \nbbP[W_0>w_0|\ncalV_0=v_0]\right)
 \\&=v_0^{\frac{1}{\alpha}}\frac{\prod\limits_{k=1}^{K}\overline{F}_{R_k}\left(\pow{0}{k}w_0\right)f_{Y_0}(v_0^{\frac{1}{\alpha}}w_0)}{\ncalA_{0|v_0}}.
 \end{align*}
 For \textbf{case 2:} when $j\in\ncalK$, $\nbbP[W_j>w_j|v_0]=\frac{1}{\ncalA_{j|v_0}}\nbbP(v_0^{-\frac{1}{\alpha}}Y_0>\pow{j}{0}\mod_j)\prod_{k=1\ k\neq j}^{K}\nbbP(P_j  \mod_j^{-\alpha}>P_k \mod_k^{-\alpha}|R_j>w_j).
 $
The rest of the proof continues in the same line of case 2 in Lemma~\ref{lemma::association_prob_sh}. 
 \subsection{Proof of Corollary~\ref{cor::f_dist_thomas}}
 \label{app::f_dist_thomas}
 The serving distance distribution when the user is served by its own cluster center is
\begin{align*}
\fservcond{0}(w_0|v_0)&=\frac{v_0^{\frac{1}{\alpha}}}{\ncalA_{0|v_0}}\prod\limits_{k=1}^K\overline{F}_{R_k}\left(\pow{0}{k}w_0\right)f_{Y_0}(v_0^{\frac{1}{\alpha}}w_0).
\end{align*}
Substituting $\overline{F}_{R_k}(\pow{0}{k}w_0)$ from Eq.~\ref{eq::nn_distribution_ccdf} and $f_{Y_0}(v_0^{\frac{1}{\alpha}}w_0)$ from Eq~\ref{eq::dist_thomas}
\begin{align}
 \fservcond{0}(w_0|v_0)&=\frac{v_0^{\frac{1}{\alpha}}}{\ncalA_{0|v_0}}\prod\limits_{k=1}^K\exp\left(-\pi\sum\limits_{k=1}^K\lam_k\pow{0}{k}^2w_0^2\right)\notag\\&\qquad\times\frac{v_0^{\frac{1}{\alpha}}w_0}{\sigma_i^2}\exp\bigg(-\frac{v_0^{\frac{2}{\alpha}}w_0^2}{2\sigma_i^2}\bigg).
\end{align}
Putting $\lam_0$ as defined before, we obtain the desired result. 
 For other open access tiers except the $0^{th}$ tier we can perform similar steps to find $\fservcond{j}(w_j|v_0)$. Starting from Lemma~\ref{lemma::serving_distance_dist_sh}, 
\begin{align*}
&\fservcond{j}(w_j|v_0)=\frac{1}{\ncalA_{j|v_0}}\Scale[0.9]{\overline{F}_{Y_{0}}(v_0^{\frac{1}{\alpha}}\pow{j}{0}w_j)\prod\limits_{\substack{k=1\\k\neq j}}^K\overline{F}_{\mod_{k}}(\pow{j}{k}w_j)f_{R_j}(w_j)}\\
 &\myeq{a}\frac{1}{\ncalA_{j|v_0}}\exp\bigg(-\frac{v_0^{\frac{2}{\alpha}}\pow{j}{0}^2w_j^2}{2\sigma_i^2}\bigg)\exp\bigg(-\pi\sum\limits_{\substack{k=1\\k\neq j}}^K\lam_k\pow{j}{k}^2w_j^2\bigg)2\pi\lam_j\\&\times \exp(-\pi\lam_jw_j^2)w_j
 =\frac{2\pi\lam_j}{\ncalA_{j|v_0}}\exp\big(-{\pi\sum\limits_{k=0}^K\lam_k\pow{j}{k}^2w_j^2}\big)w_j,
\end{align*}
where (a) follows from substitution of ${f}_{R_j}(\cdot)$, $\overline{F}_{R_k}(\cdot)$, $\overline{F}_{R_0}(\cdot)$ by Eq.~\ref{eq::nn_distribution}, Eq.~\ref{eq::nn_distribution_ccdf} and Eq.~\ref{eq::dist_thomas_ccdf}.
\subsection{Proof of Theorem~\ref{th::coverage}}
\label{app::coverage}
Recalling the definition of $\pc_j^{(i)}$ in Eq.~\ref{eq::coverage_per_tier_no_sh}, we first calculate the conditional probability, $\nbbP(\sinr(W_j)>\tau|W_j=w_j,\ncalV_0=v_0)$  $\forall j \in \ncalK_{1}$. The final result can be obtained by taking expectation with respect to $W_j$ and $\ncalV_0$. For \textbf{case 1:} when $j\in\ncalK$,
\begin{align*}
&\nbbP\biggl(\frac{P_j  \h{x}{j} w_j^{-\alpha}}{N_0+\sum\limits_{k =0}^{K} \ncalI_{o(j,k)}+\ncalI_c}>\tau|\ncalV_0=v_0\biggr)\\&
    \myeq{a}\nbbE\exp\left(-\frac{\tau w_j^{\alpha}}{P_j}\bigg(N_0+\sum\limits_{k =0}^{K} \bigg(\ncalI_{o(j,k)}+\ncalI_{c(k)}\bigg)\bigg)|\ncalV_0=v_0\right)\\
    &\myeq{b}\Scale[0.90]{\exp\bigg(-\frac{\tau N_0 w_j^{\alpha}}{P_j}\bigg)\nbbE\exp\bigg(-\frac{\tau w_j^{\alpha}}{P_j}\ncalI_{o(j,0)}|\ncalV_0=v_0\bigg)}\notag\\&{ \times \nbbE\exp\bigg(-\frac{\tau w_j^{\alpha}}{P_j}\sum\limits_{k =1}^{K} \ncalI_{o(j,k)}\bigg)\nbbE\exp\bigg(-\frac{\tau w_j^{\alpha}}{P_j}\sum_{k=1}^K\ncalI_{c(k)}\bigg)}\\
&=\exp\bigg(-\frac{\tau N_0 w_j^{\alpha}}{P_j}\bigg)\ncalL_{\ncalI_{o(j,0)}}\left(\frac{\tau w_j^{\alpha}}{P_j}|v_0\right)\prod\limits_{k =1}^{K}\ncalL_{\ncalI_{o(j,k)}}\left(\frac{\tau w_j^{\alpha}}{P_j}\right)\\&\qquad\qquad\times\prod_{k=1}^{K}\ncalL_{\ncalI_{c(k)}}\left(\frac{\tau w_j^{\alpha}}{P_j}\right)   ,
 \end{align*}
where (a) follows from $h_j\sim\exp(1)$, (b) is due to the independence of the interference from open and closed access tiers. Also note that none of the interference components except $\ncalI_{o(j,0)}$ depends on $\ncalV_0$. 

\textbf{Case 2:} For $j=0$, no contribution due to $\ncalI_{o({0,0})}$ will be accounted for. Hence, 
\begin{align*}
&\nbbP(\sir(W_0)>\tau|\serv_0=w_0)=\exp\bigg(-\frac{\tau N_0 w_0^{\alpha}}{P_j}\bigg)\\&\qquad\qquad\qquad\times\prod\limits_{k =1}^{K}\ncalL_{\ncalI_{0(0,k)}}\left(\frac{\tau w_0^{\alpha}}{P_j}\right)\prod\limits_{k=1}^K\ncalL_{\ncalI_{c(k)}}\left(\frac{\tau w_0^{\alpha}}{P_j}\right). 
\end{align*}

\subsection{Proof of Lemma~\ref{lemma::L_I_0(j,k)}}
\label{app::I_o_jk}
By definition, the Laplace transform of interference is  $\ncalL_{\ncalI_{o(j,k)}}(s)=\nbbE(\exp(-s\ncalI_{o(j,k)}))$
\begin{align}&=\nbbE\bigg[\exp\bigg(-s\sum\limits_{\scriptscriptstyle\substack{\mathbf{x}_k\in \Phi_k\\ \setminus b(\mathbf{0},\pow{j}{k}\serv_j)}}P_k \h{x}{k}\|\mathbf{x}_k\|^{-\alpha} \bigg)\bigg]\notag\\
 &\myeq{a}\nbbE_{\Phi_k}\left[\prod \limits_{\substack{\mathbf{x}_k\in \Phi_k\\\setminus b(\mathbf{0},\pow{j}{k}\serv_j)}}\nbbE_{h_k}\left(\exp\left(-sP_k \h{x}{k} \|\mathbf{x}_k\|^{-\alpha}\right)\right)\right]\notag\\
&\myeq{b}\nbbE_{\Phi_k}\left[\prod \limits_{\substack{\mathbf{x}_k\in \Phi_k\\\setminus b(\mathbf{0},\pow{j}{k}\serv_j)}}\frac{1}{1+sP_k\|\mathbf{x}_k\|^{-\alpha}}\right]\notag\\
&\myeq{c}\exp\bigg(-2\pi\lam_k\int\limits_{\pow{j}{k}\serv_j}^{\infty}\left(1-\frac{1}{1+sP_kr^{-\alpha}}\right)\:r\: {\rm d}r\bigg),\label{eq::lemma_3_intermediate}
 \end{align}
 where (a) is due to the i.i.d. assumption of $h_k$, (b) follows from $h_k \sim \exp(1)$, (c) follows from the transformation to polar coordinates and probability generating functional of homogeneous PPP \cite{HaeB2013}.
The final result can be obtained by using the integral in \cite[Eq. 3.194.1]{zwillinger2014table}.
\subsection{Proof of Lemma~\ref{lemma::lemma_limit}}
\label{app::lemma_limit}
As will be evident from the proof, the limiting arguments as well as the final limit remain the same irrespective of the value of random variable $\mathcal{V}_0$. Therefore, for notational simplicity, we provide this proof for the no shadowing scenario, which is without loss of generality. The Euclidean
distance from the typical user to its cluster center is now $Z = \std Y_0$.  In the expression of $\pc^{(i)}_{j}$ in Eq.~\ref{eq::p_c_j_compact}, we are particularly interested in inner integral which comes from the Laplace transform of interference from the BS at $0^{th}$ tier derived in Lemma~\ref{lemma::I_o_j0}. From Eq.~\ref{eq::Io_j0_general}  with the substitution ${\pow{j}{i}w_j}=x$ and $v_0=1$, we can write:
\begin{align*}
&\int\limits_{x}^{\infty}\frac{1}{1+\tau(\frac{z}{x})^{-\alpha}}\frac{f_{Z}(z)}{\overline{F}_{Z}(x)}{\rm d}z=\nbbE_{Z}\left[\frac{1}{1+\tau{\left(\frac{Z}{x}\right)}^{-{\alpha}}}|Z>x\right]\\&=\nbbE_{Y_0}\left[\frac{1}{1+\tau{\left(\frac{\std Y_0}{x}\right)}^{-{\alpha}}}|Y_0>\frac{x}{\std}\right].
\end{align*}
Hence, $\lim\limits_{\std\to\infty}\int\limits_{\frac{x}{\std}}^{\infty}\frac{1}{1+\tau{\left(\frac{\std Y_0}{x}\right)}^{-{\alpha}}}\frac{f_{Y_0}(y_0)}{\overline{F}_{Y_0}(\frac{x}{\std })}{\rm d}y_0=\int\limits_{0}^{\infty}f_{Y_0}(y_0){\rm d}y_0=1.
$
Thus, as $\std\to\infty$, the inner integral tends to 1. So, $
\lim\limits_{\std\to\infty}\pc_j^{(i)}=\lam_j/{\sum\limits_{k=1}^K\pow{j}{k}^2\bigg(\lam_k(\G+1)+\lamc_k\H\bigg)}$.
Now we are left with the limit of the first term $\ncalA_0\pc_0^{(i)}$ i.e., the contribution of the $0^{th}$ tier in $\pc^{(i)}$ which will obviously go to zero as $\std\to \infty$.
{
\bibliographystyle{IEEEtran}
\bibliography{Paper_TW_Jan_04_0873_R1.bbl}
}
\end{document}